\documentclass[a4paper,onecolumn,11pt,accepted=2025-01-16]{quantumarticle}
\pdfoutput=1
\usepackage[utf8]{inputenc}
\usepackage[english]{babel}
\usepackage[T1]{fontenc}
\usepackage[dvips]{graphicx} 
\usepackage{adjustbox}
\usepackage{pifont}
\usepackage{amsthm}
\usepackage{amsfonts}
\usepackage{amssymb}
\usepackage{amscd}

\usepackage{amsmath}    
\usepackage{enumerate}
\usepackage{epsfig}
\usepackage{subfig}
\usepackage{subfloat}
\usepackage{xcolor}			
\usepackage{physics}
\usepackage{lipsum}
\usepackage[most]{tcolorbox}
\usepackage[]{qcircuit}
\usepackage{tabularx}
\usepackage{appendix}
\usepackage{hyperref}

\usepackage[numbers,sort&compress]{natbib}

\newtcolorbox[auto counter]{game}[2]{
	label=#1,
	title=Box \thetcbcounter #2,
	colback=blue!10,
	colframe=blue!50!black
}

\newtheorem{theorem}{Theorem}
\newtheorem{lemma}{Lemma}
\newtheorem{corollary}{Corollary}

\newtheorem{question}{Question}

\newcommand{\lket}[1]{\vert #1 \rangle\!\rangle}
\newcommand{\lbra}[1]{\langle\!\langle #1 \vert}
\newcommand{\lbraket}[2]{\langle\!\langle #1 \vert #2 \rangle\!\rangle}

\begin{document}

\title{Optimizing Circuit Reusing and its Application in Randomized Benchmarking}

\author{Zhuo Chen}\thanks{These authors contributed equally to this work.}
\author{Guoding Liu}\thanks{These authors contributed equally to this work.}
\author{Xiongfeng Ma}
\email{xma@tsinghua.edu.cn}
\affiliation{Center for Quantum Information, Institute for Interdisciplinary Information Sciences, Tsinghua University, Beijing, 100084 China}

\begin{abstract}
Quantum learning tasks often leverage randomly sampled quantum circuits to characterize unknown systems. An efficient approach known as ``circuit reusing,'' where each circuit is executed multiple times, reduces the cost compared to implementing new circuits. This work investigates the optimal reusing times that minimizes the variance of measurement outcomes for a given experimental cost. We establish a theoretical framework connecting the variance of experimental estimators with the reusing times $R$. An optimal $R$ is derived when the implemented circuits and their noise characteristics are known. Additionally, we introduce a near-optimal reusing strategy that is applicable even without prior knowledge of circuits or noise, achieving variances close to the theoretical minimum. To validate our framework, we apply it to randomized benchmarking and analyze the optimal $R$ for various typical noise channels. We further conduct experiments on a superconducting platform, revealing a non-linear relationship between $R$ and the cost, contradicting previous assumptions in the literature. Our theoretical framework successfully incorporates this non-linearity and accurately predicts the experimentally observed optimal $R$. These findings underscore the broad applicability of our approach to experimental realizations of quantum learning protocols.
\end{abstract}

\maketitle

\section{Introduction}\label{sc:intro}
The development of large-scale and functional quantum computers requires the efficient characterization of quantum systems, a task often referred to as ``quantum learning''~\cite{cross2015quantum,huang2022quantum, huang2020predicting}. This includes estimating observables of quantum states~\cite{paris2004quantum, adamson2010improving, Aaronson2018shadow} and assessing noise in quantum channels~\cite{wiebe2014hamiltonian, flammia2020efficient, flammia2021pauli}. A wealth of techniques within quantum learning have been proposed with the aim of efficiently evaluating the properties of quantum states and channels. Among these, the ``randomization'' approach is pivotal, where outcomes from random circuits are averaged to derive desired quantities. This approach involves sampling random quantum states~\cite{urbina2013random_state}, employing random gate sequences~\cite{RB2005, RB2008, Classical_shadow}, and implementing randomized measurements~\cite{Elben2018Random, Tiff2019randomized, elben2023randomized_meas}.

Within the context of quantum learning tasks, minimizing resource consumption, such as the number of implemented circuits and experimental duration, is crucial. A viable strategy to achieve this is ``circuit reusing'', where each randomly sampled circuit is used multiple times, including state preparation, gate implementation, and measurement. This strategy is extensively applied in the experimental execution of quantum learning tasks, such as classical shadow \cite{Huggins2022} and randomized benchmarking \cite{RB_with_confidence}. In experiments, initializing a new random circuit generally consumes more time than reusing an existing one due to necessary steps like random sampling, quantum gate decomposition, and other hardware preparations. Therefore, employing a circuit reusing strategy can reduce experimental costs while maintaining a reasonable data sample volume.

Theoretically, increasing the number of different circuits can enhance randomness and, thereby, improve estimation precision \cite{huang2020predicting}. This introduces a trade-off between cost savings and improved estimation accuracy. An important question emerges: what is the optimal number of times a circuit should be reused to minimize the estimation variance for a given experimental cost? This issue has been explored in the context of classical shadows~\cite{Thrifty_shadow, zhou2023performance}, analyzed in a noiseless scenario, with a cost model suggesting that executing a circuit $R$ times incurs a cost $t(R)=R+\alpha$, where $\alpha$ represents the fixed overhead associated with initializing a new circuit.

When considering the reusing strategy in broader quantum learning tasks and practical scenarios, several challenges arise. Firstly, unavoidable circuit noise must be considered in the assessment of the optimal reusing times, particularly for tasks aimed at noise benchmarking~\cite{RB2008,RB2012}. Furthermore, experimental data indicates that the assumption of a constant overhead $\alpha$ is unrealistic, necessitating a more universal cost model applicable to diverse quantum learning tasks.

In this study, we examine how the number of reusing times, $R$, influences the variance of outcomes in diverse quantum learning tasks. We introduce a cost model, $T(N,R) = Nt(R)$, to describe the cost of implementing $N$ distinct circuits, each reused $R$ times, where $t(R)$ varies as a function of $R$. With this model, we formulate an optimization issue that aims to minimize the variance of the target estimator as a function of the reusing times $R$ for a fixed experimental budget.

Employing the Law of Total Variance, we decompose the variance of the estimator into two components and establish the relationship between these components and the reusing times $R$. This analysis enables us to identify the optimal $R$ that minimizes variance for various cost models $t(R)$ and across different quantum learning tasks. While the value of the optimal reusing times depends on the experimental configuration, implemented circuits, and their specific noise channel characteristics, we introduce a practical ``near-optimal'' solution that requires no prior knowledge of the noise channel or the executed circuits. We prove that the variance of the near-optimal solution does not exceed a constant multiple of the minimal achievable variance given a bound on $t(R)$. Notably, both the optimal and the near-optimal strategies are versatile and applicable to various quantum learning tasks such as randomized benchmarking (RB) \cite{RB2005, RB2008, RB2012}, shadow estimation \cite{Classical_shadow}, direct fidelity estimation \cite{da2011practical}, and quantum tomography \cite{chuang1997prescription, poyatos1997complete}. Additionally, the principles underlying this solution can enhance other techniques that utilize the ``circuit reusing'' strategy, including quantum error mitigation \cite{cai2023quantum_error_mitigation} and entanglement detection \cite{entanglement}.

As an application of our findings, we explore the reusing strategy within the standard RB protocol \cite{RB2012} and identify the optimal reusing times $R^*$ with numerical simulations across typical noise channels. Moreover, we conduct $2$-qubit standard RB experiments with different reusing times $R$ ranging from $2$ to $20000$ on a superconducting quantum computing platform~\cite{Shaowei2022CZ}. Our experimental findings reveal that the simplest linear cost model $t(R)=R+\alpha$ does not hold. This observation leads us to develop a tailored cost model specific to our experimental setup. In experiments, the optimal reusing times $R^*$ ranges from $100$ to $500$, aligning closely with our theoretically predicted optimal solution of $200$. Furthermore, our near-optimal solution, with $R_0=333$, exhibits a variance below $2$ times the minimal variance observed experimentally, confirming our theoretical predictions.

The structure of this paper is organized as follows. Section \ref{sec:2opt} introduces our analytical framework and explores the optimization of reusing times $R$ under various experimental conditions. Section \ref{sc:theory} applies our findings to RB and presents corresponding numerical simulations across various noise channels. Section \ref{sc:Exp} details the experimental demonstration conducted on a superconducting quantum computing platform. Finally, we conclude in Section~\ref{sc:conclude} with a summary of our results and implications for future research.

\section{Optimize circuit reusing in quantum learning tasks}\label{sec:2opt}
\subsection{Formulation of the optimization issue for circuit reusing times}
The task of quantum learning aims to assess the property of quantum states or quantum channels via measurement results from multiple circuits. When the protocol of the target task involves one-shot measurement across distinct circuits, it is feasible to perform repeated measurements to reuse circuits. Below, we analyze the relationship between the estimation variance and the reusing times.

As depicted in Fig.~\ref{fig:flowchart}, let us outline the steps involved in a quantum learning task: We execute $N$ random circuits, with each circuit potentially involving distinct state $\rho_i$, gate sequence $\vec{G}_i$, and measurement $Q_i$. Implementation of each circuit is repeated $R$ times. For the $i$-th circuit, the $r$-th single-shot measurement result is represented as a random variable, denoted by $X_r(\rho_i, Q_i, \vec{G}_i)$. The target estimator $\Bar{X}_N(R)$ is then defined as
\begin{equation}
\begin{split}
\Bar{X}_N(R) &= \frac{1}{N} \sum_i \Bar{X}_R(\rho_i, Q_i, \vec{G}_i),\\
\Bar{X}_R(\rho_i, Q_i, \vec{G}_i) &= \frac{1}{R}\sum_r X_r(\rho_i, Q_i, \vec{G}_i).
\end{split}
\end{equation}
More generally, $X_r(\rho_i, Q_i, \vec{G}_i)$ can be a value of the $r$-th single-shot measurement result after classical post-processing. For our further derivation, we only require the independent and identically distributed condition for $X_r(\rho_i, Q_i, \vec{G}_i)$ over different $r$ and for $\Bar{X}_R(\rho_i, Q_i, \vec{G}_i)$ over different random gate sequences. This is ensured by the uniform and random sampling for gate sequences and the independence among different single-shot measurement results.

\begin{figure}[htbp!]
\centering
\includegraphics[width=15cm]{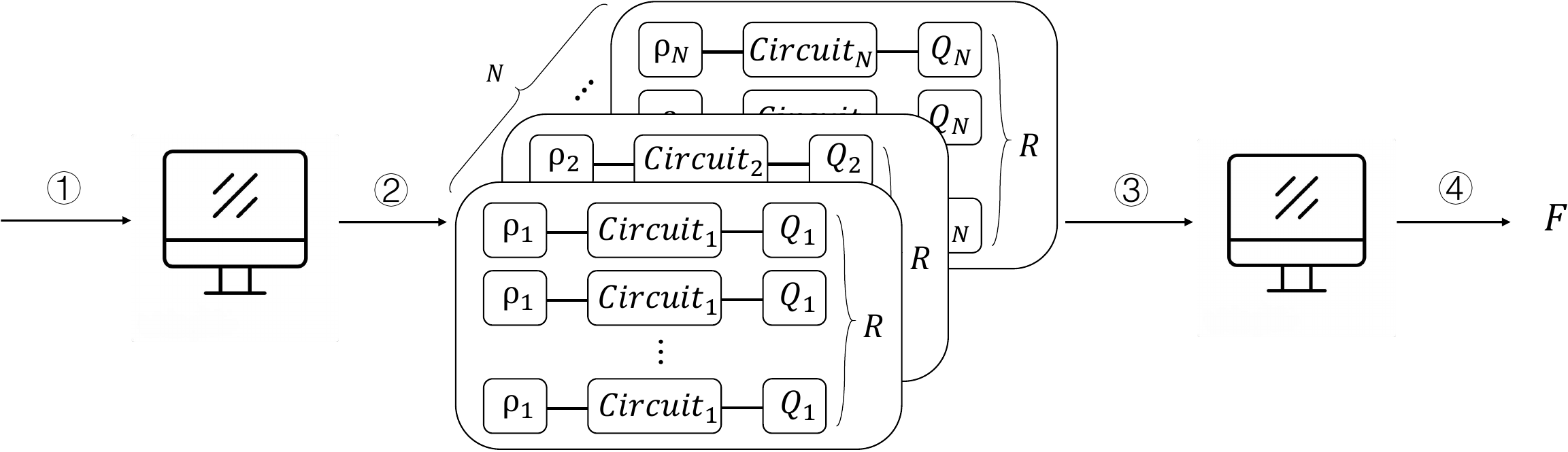}
\caption{Procedure of a quantum learning task utilizing circuit reusing: \large{\textcircled{\small{1}}}\normalsize Input the quantum learning objective and the prior knowledge about quantum device performance and implemented circuits to a classical computer for analysis. \large{\textcircled{\small{2}}}\normalsize Design the learning protocol, setting parameters such as the reusing times $R$ and the number of unique circuits $N$, and generate the corresponding quantum circuits. \large{\textcircled{\small{3}}}\normalsize Execute quantum operations and measurements to gather data for estimating the target quantity. \large{\textcircled{\small{4}}}\normalsize Analyze the collected data using classical post-processing to derive the desired quantity.}
\label{fig:flowchart}
\end{figure}

To identify the optimal value of $R$ to minimize the variance of $\Bar{X}_N(R)$, we develop the subsequent optimization model. Suppose implementing each circuit $R$ times consumes cost $t(R)$, which typically pertains to time expense, the total cost $T_0$ for acquiring $NR$ samples can be described as
\begin{equation}\label{alpha}
T_0=Nt(R).
\end{equation}
The situation can further be generalized to the case that $t(R)$ also depends on $N$. With this foundation, we formulate the optimization question as follows.

\begin{question}\label{ques:opt}
Given a predetermined total cost $T_0$, the objective is to determine the optimal reusing times $R$ that minimizes the variance of the estimator $\Bar{X}_N$. Mathematically,
\begin{equation}
\begin{split}
&\mathop{\arg\min}\limits_{R} \mathbb{V}(\Bar{X}_N) \\
&s.t.\ Nt(R) = T_0.
\end{split}
\end{equation}
\end{question}

In previous works~\cite{Thrifty_shadow, zhou2023performance}, an optimization issue was analyzed for the strategy of circuit reusing within the framework of randomized shadow estimation, assuming a linear cost model of $t(R) = R + \alpha$, where $\alpha$ is a positive constant representing the extra cost incurred when initializing a new circuit compared to reusing an existing one. In contrast, our approach accommodates more general scenarios. Observations from RB experiments suggest that the linear model may not be applicable. This variability is substantiated by runtime data from our RB experiments conducted on a superconducting quantum computing platform, which we detail in Section \ref{sc:Exp}. Consequently, our optimization issue, as outlined in Question~\ref{ques:opt}, requires a more universal expression for $t(R)$.

Assuming that the circuits $\{\rho_i, Q_i, \vec{G}_i\}$ are uniformly and independently determined, the variance of the estimator is given by:
\begin{equation}
\mathbb{V}(\Bar{X}_N(R))=\frac{1}{N}\mathbb{V}(\Bar{X}_R(\rho_i, Q_i, \vec{G}_i)),
\end{equation}
where $\Bar{X}_R$ is the average over $R$ repeated measurements. Applying the Law of Total Variance, we decompose the variance of $\Bar{X}_R$ as follows:
\begin{equation} \label{proj_measurement}
\begin{split}
\mathbb{V}(\Bar{X}_R) &= \mathbb{E}_{i}(\mathbb{V}(\Bar{X}_R|(\rho_i, Q_i, \vec{G}_i))) + \mathbb{V}_{i}(\mathbb{E}(\Bar{X}_R|(\rho_i, Q_i, \vec{G}_i)))\\
&= \frac{1}{R}\mathbb{E}_{i}(\mathbb{V}(X_r|(\rho_i, Q_i, \vec{G}_i))) + \mathbb{V}_{i}(\mathbb{E}(X_r|(\rho_i, Q_i, \vec{G}_i))).
\end{split}
\end{equation}
Thus, the variance of $\Bar{X}_N(R)$ is expressed as:
\begin{equation}\label{eq:variance}
\begin{split}
\mathbb{V}(\Bar{X}_N(R))=\frac{t(R)}{T_0}\left(\frac{Y}{R}+Z\right),
\end{split}
\end{equation}
where \begin{equation}\label{eq:coefYZ}
\begin{split}
Y&=\mathbb{E}_{i}(\mathbb{V}(X_r|(\rho_i, Q_i, \vec{G}_i))),\\
Z&=\mathbb{V}_{i}(\mathbb{E}(X_r|(\rho_i, Q_i, \vec{G}_i))).
\end{split}
\end{equation}
Both coefficients $Y$ and $Z$, related to the statistical properties of the measurement result $X_r$, remain constant relative to $R$. In the subsequent analysis, we determine the optimal reusing times $R$ based on whether the coefficients $Y$ and $Z$ are known or unknown.

\subsection{Optimizing circuit reusing with circuit and noise information}\label{subsec:R*}
When detailed information about the noise channel and the implemented circuits is available, it becomes feasible to precisely determine the coefficients $Y$ and $Z$ and subsequently calculate the optimal reusing times.
Assuming the cost of implementing each circuit $R$ times follows the model $t(R) = \alpha + \beta R$, where $\alpha$ and $\beta$ are positive constants. This cost model, referred to as the ``constant-cost model,'' indicates that initializing a new circuit incurs an additional cost $\alpha$, whereas reusing an existing one costs $\beta$ per reuse. Under this model, the optimal reusing times, denoted by $R^*$, can be determined as follows.

\begin{theorem}\label{thm:optimalR}
If the cost of implementing a circuit $R$ times is modeled as $t(R) = \alpha+\beta R$, with $\alpha,\beta $ being constant, the optimal reusing times $R^*$ that minimizes the variance is given by
\begin{equation}\label{eq:R*YZ}
R^* = \sqrt{\frac{\alpha Y}{\beta Z}}.
\end{equation}
\end{theorem}

The above theorem can be proved by an inequality:
\begin{equation}\label{eq:optvariance}
\begin{split}
\mathbb{V}(\Bar{X}_N(R)) &= \frac{\alpha+\beta R}{T_0}\left(\frac{Y}{R}+Z\right) \\
&\geq \frac{1}{T_0}\left(\sqrt{\beta Y}+\sqrt{\alpha Z}\right)^2,
\end{split}
\end{equation}
where the equality holds when $R = \sqrt{\alpha Y/\beta Z}$. 
From Eq.~\eqref{eq:optvariance}, it is clear that the improvement offered by the optimal strategy, compared to naively choosing $R = 1$, depends on the specific values of $Y$ and $Z$. In the extreme case, for example when $Y\gg Z$ and $\alpha Z\gg \beta Y$, the speedup factor of optimal strategy is given by:
\begin{equation}
\begin{split}
    \frac{\mathbb{V}(\Bar{X}_N(R=1))}{\mathbb{V}(\Bar{X}_N(R^*))}=\frac{(\alpha+\beta)(Y+Z)}{\beta Y+\alpha Z+1\sqrt{\alpha\beta YZ}}\approx \frac{Y}{Z}\gg 1.
\end{split}
\end{equation}
Besides, since the reusing times must be an integer, the actual optimal reusing times takes value \begin{equation}
R^*=\mathop{\arg\min}\limits_{R\in\{\lceil\sqrt{\frac{\alpha Y}{\beta Z}}\rceil,\lfloor\sqrt{\frac{\alpha Y}{\beta Z}}\rfloor\}}\mathbb{V}\left(\Bar{X}_N\left(R\right)\right).
\end{equation}
In future discussions, we will default on this point without ambiguity.

The conclusion of Theorem~\ref{thm:optimalR} can be extended to other practical cases. Here, as an example, we take the form of $t(R)$ as
\begin{equation}
t(R) = C_1\lceil R/R_c\rceil+C_2,
\end{equation}
which is a time cost model that we observed in experiments as elaborated in Section~\ref{sc:Exp}. The terms $C_1$ and $C_2$ are both positive. In this case, the variance is given by
\begin{equation}
\mathbb{V}(\Bar{X}_N(R)) = \frac{C_1\lceil R/R_c\rceil+C_2}{T_0}\left(\frac{Y}{R}+Z\right).
\end{equation}
A direct calculation indicates that in this case, the optimal reusing times takes value
\begin{equation}
R^*=\mathop{\arg\min}\limits_{R\in\{\lceil\sqrt{\frac{C_2Y}{C_1ZR_c}}\rceil R_c,\lfloor\sqrt{\frac{C_2Y}{C_1ZR_c}}\rfloor R_c\}}\mathbb{V}\left(\Bar{X}_N\left(R\right)\right).
\end{equation}
More general cases of $t(R)$, including only the bound instead of concrete forms of $t(R)$ known, are discussed in Appendix~\ref{appendsc:opt}.


\subsection{Optimizing circuit reusing without circuit and noise information}
While the above analyses show that the solution to Question~\ref{ques:opt} depends on the values of $Y$ and $Z$, the two coefficients depend on both the distribution of the executed circuits and the specific noise characteristics of quantum operations. In practice, fully characterizing noise can be challenging, rendering the precise values of $Y$ and $Z$ often unavailable and making it difficult to pinpoint the exact optimal reusing times $R^*$. However, the cost model $t(R)$, typically determined by analyzing experimental runtime, is usually accessible. With this information at hand, it is feasible to choose a value of $R$ that ensures the variance remains below a threshold.

Let us consider a scenario where we have prior knowledge of the cost bounds for implementing circuits multiple times:
\begin{equation}\label{eq:costbound}
\alpha_l+\beta_l R \leq t(R) \leq \alpha_u+\beta_u R,\ \forall R \in \mathbb{Z}^+,
\end{equation}
where $\alpha_l$, $\alpha_u$, $\beta_l$, and $\beta_u$ are positive constants, indicating that the actual cost model $t(R)$ falls between two such constant-cost models. Given this bounded cost model, we propose a near-optimal solution for Question \ref{ques:opt}, which does not require prior knowledge of circuit specifics or noise characteristics.




\begin{theorem}\label{thm:2opt}
In a quantum learning task that averages across various circuits represented as $\{\rho_i, Q_i, \vec{G}_i\}$, the calculated value $\Bar{X}_N(R)=\frac{1}{NR} \sum_{i,r} X_r(\rho_i, Q_i, \vec{G}_i)$ is derived after reusing each of the $N$ distinct circuits $R$ times, where $X_r$ represents the $r$-th one-shot measurement result. Assume that the total time cost satisfies $T_0 = Nt(R)$ with $\alpha_l+\beta_l R \leq t(R) \leq \alpha_u+\beta_u R$. Within fixed total time cost, we have 
\begin{equation}
\mathbb{V}(\Bar{X}_N(R_0))\leq \left(\frac{\alpha_u}{\alpha_l}+\frac{\beta_u}{\beta_l}\right)\mathbb{V}(\Bar{X}_N(R^*)),
\end{equation}
where $R_0=\frac{\alpha_l}{\beta_l}$, and $R^*$ represents the optimal reusing times that minimizes $\mathbb{V}(\Bar{X}_N(R))$.
\end{theorem}
\begin{proof}
Through the variance form in Eq.~\eqref{eq:variance}, we obtain that 
\begin{equation}
\begin{split}
\mathbb{V}\left(\Bar{X}_N(R^*)\right)&\geq\min_R\left(\frac{\beta_l R+\alpha_l}{T_0}\left(\frac{Y}{R}+Z\right)\right)\\
&\geq \frac{\beta_l Y+\alpha_l Z}{T_0}.\\
\mathbb{V}\left(\Bar{X}_N\left(\frac{\alpha_l}{\beta_l}\right)\right)&\leq \frac{\beta_u \frac{\alpha_l}{\beta_l}+\alpha_u}{T_0}\left(\frac{Y\beta_l}{\alpha_l}+Z\right)\\
&= \left(\frac{\alpha_u}{\alpha_l}+\frac{\beta_u}{\beta_l}\right)\frac{\beta_l Y+\alpha_l Z}{T_0},\\
\end{split}
\end{equation}
By comparing the two variances, we prove our result.
\end{proof}

The following text refers to the near-optimal solution as $R_0$. When we return to the case of the constant-cost model, it is straightforward to develop a relaxed version as outlined below.

\begin{corollary}\label{cor:2opt_constant}
When the cost to implement each circuit $R$ times is given by $t(R) = \alpha + \beta R$ with $\alpha$ and $\beta$ being positive constants, 
$R_0=\frac{\alpha}{\beta}$ deems as a $2$-optimal solution to Question~\ref{ques:opt}, i.e. \begin{equation}
\mathbb{V}(\Bar{X}_N(R_0))\leq 2\mathbb{V}(\Bar{X}_N(R^*)).
\end{equation}
\end{corollary}

Note that the value of the optimal $R^*$ may significantly differ from the $\left(\frac{\alpha_u}{\alpha_l}+\frac{\beta_u}{\beta_l}\right)$-optimal value $R_0=\frac{\alpha_l}{\beta_l}$ in certain scenarios. Nonetheless, the balanced strategy $R_0$ in Theorem~\ref{thm:2opt} ensures that the overall variance remains within acceptable limits, which is particularly advantageous in situations where we have no information about the form of the noise channel. More importantly, the applicability of this theorem encompasses various quantum learning and information processing methodologies employing random circuit and reusing techniques, even containing cases where the underlying sampled circuits are unknown~\cite{Elben2018Random, Tiff2019randomized}.

\section{Optimal reusing times in standard RB}\label{sc:theory}
\subsection{Standard RB with circuit reusing}
Here, we apply our results on optimal reusing times to RB~\cite{RB2005}, an efficient tool for assessing the quality of quantum gates, robust against state preparation and measurement (SPAM) errors. We focus on the standard RB protocol~\cite{RB2008, RB2012}, which evaluates the average fidelity of gates in a $2$-design gate set such as the Clifford group~\cite{dankert2009exact, RB2012}. The fidelity of a gate $G$, denoted as $\mathcal{G}(\rho) = G\rho G^\dagger$, is defined by:
\begin{equation}
\begin{split}
F_{avg}(\Tilde{\mathcal{G}}, \mathcal{G}) := \int d\phi \, \Tr(\mathcal{G}(|\phi\rangle \langle \phi|)\Tilde{\mathcal{G}}(|\phi\rangle \langle \phi|)),
\end{split}
\end{equation}
where $\Tilde{\mathcal{G}}$ is the noisy implementation of $\mathcal{G}$ and $d\phi$ represents the uniform Haar measure over pure quantum states.

In RB protocols, the noise models of quantum gates are typically assumed to be Markovian and gate-independent~\cite{general}. Here, we adhere to this assumption, suggesting that the practical implementation of a gate $\mathcal{G}_i$ from the gate set $\mathbf{G}$ is represented by $\tilde{\mathcal{G}_i} = \Lambda\mathcal{G}_i$, where $\Lambda$ is a completely positive trace-preserving map acting as post-gate noise. Given that $\mathcal{G}$ is unitary, $F_{avg}(\Tilde{\mathcal{G}}, \mathcal{G})$ equals $F_{avg}(\Lambda, \mathcal{I})$, and we denote it as $F_{avg}(\Lambda)$ in subsequent discussions. Additional details on the RB process and related mathematical tools are presented in Appendix~\ref{app_Rep}.

In short, to estimate the average fidelity of the gate set $\mathbf{G}$ by standard RB, we sample $N$ different gate sequences $\vec{G_i}=\{G_1,G_2,\ldots,G_m\}$ with length $m$, together with a global inverse gate $G_{inv}=G_1^\dagger G_2^\dagger\ldots G_m^\dagger$ for each gate sequence. Then we apply these gates to an initial state $\rho$, and implement positive operator-valued measure (POVM) measurement $\{Q,\mathbb{I}-Q\}$ on the final state \begin{equation}\label{final_state}
\rho_f=\mathcal{G}_{inv}\Lambda\mathcal{G}_m\ldots\mathcal{G}_2\Lambda\mathcal{G}_1(\rho).
\end{equation}
Here, we adopt the convention of ignoring the noise channel of the inverse gate, as in other RB literature. This change does not influence the analysis of obtaining fidelity. One can view this omitted noise channel to be absorbed by measurements.
With circuit reusing strategy, we repeat implementing the circuit $(\rho,Q,\vec{G}_{i})$ $R$ times with fixed $\rho$ and $Q$. The $r$-th measurement result can be represented by a random variable \begin{equation}\label{eq:RB_rv}
X_r(m,\vec{G}_i)=\begin{cases}
0,&\text{w.p. }1-\Tr(Q\rho_f)\\
1,&\text{w.p. }\Tr(Q\rho_f)
\end{cases}.
\end{equation} 
Averaging over $R$ results, we obtain the survival probability associated with $\vec{G_i}$,
\begin{equation}
\Bar{X}_R(m,\vec{G}_i)=\frac{1}{R}\sum_{r=1}^RX_r(m,\vec{G}_i).
\end{equation}
Taking the average for $N$ different gate sequences, we obtain our final estimator, \begin{equation}
\Bar{X}_N(m)=\frac{1}{N}\sum_{i=1}^N\Bar{X}_R(m,\vec{G}_i).
\end{equation} 
The flowchart is shown in Fig.~\ref{RB_flowchart}.

\begin{figure}[htbp]
\centering
\includegraphics[width=10.2cm]{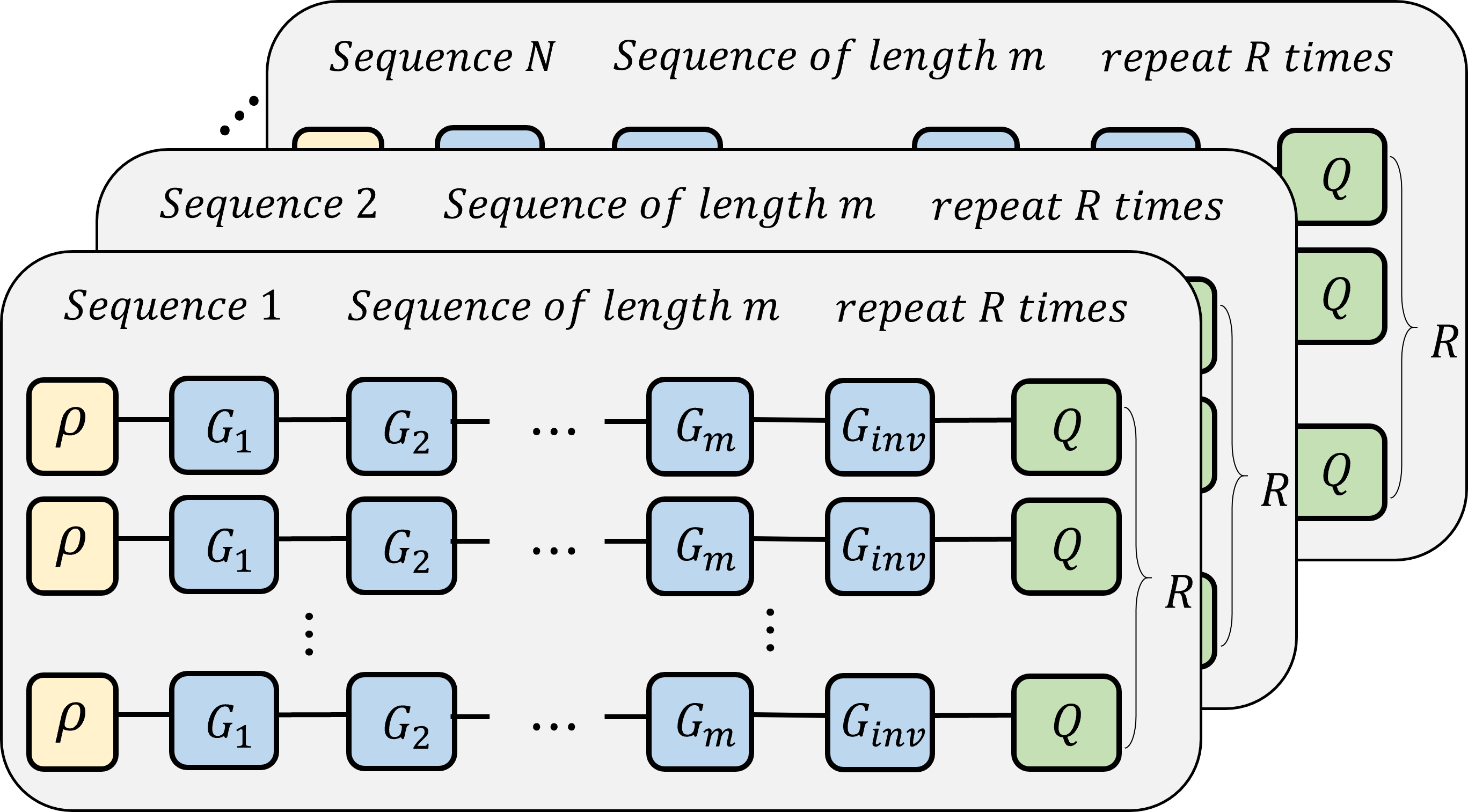}
\caption{The flowchart illustrates the process for the standard RB protocol with quantum circuits reusing. Each layer represents a gate sequence repeating $R$ times, and it requires $N$ different sequences. 
}
\label{RB_flowchart}
\end{figure}

Utilizing the $2$-design property~\cite{dankert2009exact,RB2012} of the target gate set in standard RB, the expectation value of the estimator $\Bar{X}_N(m)$ has a single exponential form, \begin{equation}
\mathbb{E}(\Bar{X}_N(m))=af(\Lambda)^m+b,
\end{equation}
where $a$ and $b$ are coefficients that absorb SPAM errors, $f(\Lambda)$ represents the quality parameter associated with the target average fidelity, \begin{equation}\label{eq:quality}
F_{avg}(\Lambda)=f(\Lambda)+\frac{1-f(\Lambda)}{d},
\end{equation}
where $d=2^n$ represents the system size.
By assigning different sequence lengths $m$, one could obtain quality parameter $f(\Lambda)$ by fitting, thus assessing the average fidelity $F_{avg}(\Lambda)$.

\subsection{Theoretical analyses}
Note that the variance of the fidelity estimation is positively correlated with the variance of the estimator $\Bar{X}_N(m)$ \cite{epstein2014investigating}, hence one can optimize the variance $\mathbb{V}(\Bar{X}_N(m))$ to gain more accurate estimation for fidelity. 
Based on Theorem~\ref{thm:optimalR}, to obtain the optimal reusing times $R^*$ to minimize $\mathbb{V}(\Bar{X}_N(m))$, we only need to evaluate the terms $Y$ and $Z$ in Eq.~\eqref{eq:coefYZ} according to the form of the measurement outcome $X_r$. In standard RB, $X_r$ takes 1 with probability $\Tr(Q\mathcal{G}_{inv}\Lambda\mathcal{G}_m\ldots\mathcal{G}_2\Lambda\mathcal{G}_1(\rho))$ and 0 otherwise, as shown in Eq.~\eqref{eq:RB_rv}. Applying Theorem~\ref{thm:optimalR} to standard RB, we obtain the following corollary.

\begin{corollary} \label{lemma1}
In standard RB with circuit reusing, the values of $Y$ and $Z$ as defined in Eq.~\eqref{eq:coefYZ} are computed by:
\begin{equation}
\begin{split}
Y &= A - B \\
Z &= B - A^2
\end{split}
\end{equation}
where:
\begin{equation}
\begin{split}
A&:=\mathbb{E}_{\vec{G}}\left(\lbra{Q}\mathcal{G}_{inv}\Lambda\mathcal{G}_m\ldots\mathcal{G}_2\Lambda\mathcal{G}_1\lket{\rho}\right)\\
&=f(\Lambda)^m\Tr(Q\rho)+\frac{1-f(\Lambda)^m}{d}\Tr(Q)\\
&=\frac{d\cdot\Tr(Q\rho)-\Tr(Q)}{d}\cdot \left(\frac{d\cdot F_{avg}(\Lambda)-1}{d-1}\right)^m+\frac{1}{d}\Tr(Q),\\
B&:=\mathbb{E}_{\vec{G}}\left(\lbra{Q}\mathcal{G}_{inv}\Lambda\mathcal{G}_m\ldots\mathcal{G}_2\Lambda\mathcal{G}_1\lket{\rho}^2\right).
\end{split}
\end{equation}
Furthermore, when the cost of implementing a circuit $R$ times is given by $t(R) = \alpha + \beta R$ with positive constants $\alpha$ and $\beta$, the optimal value of reusing times is determined by:
\begin{equation}
R^* = \sqrt{\frac{\alpha (A - B)}{\beta (B - A^2)}}.
\end{equation}
\end{corollary}
\begin{proof}
By direct substitution of Eq.~\eqref{eq:RB_rv} into $Y$ and $Z$, we have 
\begin{equation}
\begin{split}
\mathbb{E}(X_r|\vec{G}_i)&=\lbra{Q}\mathcal{G}_{inv}\Lambda\mathcal{G}_m\ldots\mathcal{G}_2\Lambda\mathcal{G}_1\lket{\rho}\\
&=\Tr(Q\rho_f).
\end{split}
\end{equation}
\begin{equation}
\begin{split}
\mathbb{V}(X_r|\vec{G}_i)&=\mathbb{E}(X_r|\vec{G}_i)-\mathbb{E}(X_r|\vec{G}_i)^2\\
&=\Tr(Q\rho_f)-(\Tr(Q\rho_f))^2.
\end{split}
\end{equation}
\begin{equation}
\begin{split}
Y &= \mathbb{E}_i(\mathbb{V}(X_r|\vec{G}_i))\\
&= \mathbb{E}_i(\Tr(Q\rho_f)-(\Tr(Q\rho_f))^2)\\
&= \mathbb{E}_i(\Tr(Q\rho_f))-\mathbb{E}_i((\Tr(Q\rho_f))^2)\\
&= A - B.
\end{split}
\end{equation}
\begin{equation}
\begin{split}
Z &= \mathbb{V}_i(\mathbb{E}(X_r|\vec{G}_i))\\
&= \mathbb{V}_i(\Tr(Q\rho_f))\\
&= \mathbb{E}_i((\Tr(Q\rho_f))^2)  - (\mathbb{E}_i(\Tr(Q\rho_f)))^2\\
&= B - A^2.
\end{split}
\end{equation}
Here,
\begin{equation}
A = \mathbb{E}_i(\Tr(Q\rho_f)),
\end{equation}
and
\begin{equation}
B = \mathbb{E}_i((\Tr(Q\rho_f))^2).
\end{equation}
Hence,
\begin{equation}
\begin{split}
\mathbb{V}(\Bar{X}_R(m, \vec{G}_i))&= \frac{Y}{R}+Z\\
&=\frac{A-B}{R}+B-A^2.\\
\end{split}
\end{equation}

Considering that the circuit sequences $\vec{G_i}$ are sampled uniformly and independently, it follows that $\mathbb{V}(\Bar{X}_N(m,\vec{G_i}))=\frac{1}{N}\mathbb{V}(\Bar{X}_R(m,\vec{G}_i))$. When combined with the Eq.~\eqref{alpha}, we have 
\begin{equation}
\begin{split}
\mathbb{V}(\Bar{X}_N(m)) &= \frac{\alpha+\beta R}{T_0}\mathbb{V}(\Bar{X}_R(m, \vec{G}_i))\\
&=\frac{1}{T_0}\left(\beta(A-B)+\alpha(B-A^2)+\beta R(B-A^2)+\frac{\alpha(A-B)}{R}\right) \label{eq:SRBvariance}\\
&\geq \frac{1}{T_0}\left(\sqrt{(\beta(A-B)}+\sqrt{\alpha(B-A^2)}\right)^2,
\end{split}
\end{equation}
with
\begin{equation}\label{eq:R*}
R^*=\sqrt{\frac{\alpha(A-B)}{\beta(B-A^2)}},
\end{equation}
to get the minimum variance. 

As for $A$, utilizing the $2$-design property of gate set $\mathbf{G}$, we can rewrite $A$ as \begin{equation}
\begin{split}        A&=\mathbb{E}_{\vec{G}}\left(\lbra{Q}\left(\mathcal{G}_{inv}\Lambda\mathcal{G}_m\ldots\mathcal{G}_1\right)\ldots\left((\mathcal{G}_2\mathcal{G}_1)^\dagger\Lambda(\mathcal{G}_2\mathcal{G}_1)\right)(\mathcal{G}_1^\dagger\Lambda\mathcal{G}_1)\lket{\rho}\right)\\
&=\lbra{Q}\mathbb{E}\prod_{i=1}^m(\mathcal{G}_i^\dagger\Lambda\mathcal{G}_i)\lket{\rho}\\
&=\Tr(Q\Lambda_{dep}^m(\rho)),
\end{split}
\end{equation}
where $\Lambda_{dep}$ is the depolarizing channel with the same fidelity as $\Lambda$ \begin{equation}
\Lambda_{dep}(\rho)=f(\Lambda)\rho+\frac{1-f(\Lambda)}{d}\mathbb{I},
\end{equation}
and $f(\Lambda)$ is the quality parameter associated with the average fidelity of $\Lambda$ as Eq.~\eqref{eq:quality}. Further direct calculations complete the proof.
\end{proof}

To analyze the explicit formula of $B$, we rewrite it as
\begin{equation}
\begin{split}
B&=\mathbb{E}_{\vec{G}}(\Tr(Q\rho_f)^2)\\
&=\mathbb{E}_{\vec{G}}\left(\lbra{Q^{\otimes2}}\mathcal{G}_{inv}^{\otimes2}\Lambda^{\otimes2}\mathcal{G}_m^{\otimes2}\ldots\mathcal{G}_2^{\otimes2}\Lambda^{\otimes2}\mathcal{G}_1^{\otimes2}\lket{\rho^{\otimes2}}\right)\\
&=\lbra{Q^{\otimes2}}\mathcal{T}(\Lambda^{\otimes2})^m\lket{\rho^{\otimes2}},
\end{split}
\end{equation}
where \begin{equation}
\mathcal{T}(\Lambda^{\otimes2})=\mathbb{E}_{G\in\mathbf{G}}\left(\mathcal{G}_i^{\dagger\otimes2}\Lambda^{\otimes2}\mathcal{G}_i^{\otimes2}\right)
\end{equation}
is a $twirling$ of quantum channel $\Lambda^{\otimes2}$ over the two-copy representation of $\mathbf{G}$. 
When the target gate set is the $n$-qubit Clifford group $\mathbf{C}_n$, knowing the structure of the two-copy representation of the multi-qubit Clifford group $\mathcal{G}^{\otimes2}$~\cite{multi-qubit_Clifford_representations} allows us to decompose  $\mathcal{T}(\Lambda^{\otimes2})$ into several terms. One can then express $B$ as a linear combination of exponential-decay functions with respect to $m$. Since there exist equivalent representations in the irreducible decomposition of the two-copy representation of the multi-qubit Clifford group, $B$ includes matrix exponential-decay terms for general noise channel $\Lambda$ \cite{multi-qubit_Clifford_representations}. In the next subsection, we numerically evaluate $B$ and the optimal reusing times for different noise models.

It is worth noting that the result for $R^*$ described in Corollary \ref{lemma1} can be easily generalized to the variations of RB protocols. For a given RB protocol, once detailed procedures and the benchmarking group are determined, one can always represent the $r$-th time measurement result under gate sequence $\Vec{G_i}$ as a random variable $X_r(m,\vec{G_i})$. Then one can further obtain the value of $Y$, $Z$, and optimal reusing times $R^*$ by substituting $X_r(m,\vec{G_i})$ in Eq.~\eqref{eq:coefYZ}.
Interestingly, for a simple variation of standard RB, named difference RB~\cite{helsen2019UsingFewSamples}, we can derive an analytical lower bound of the optimal reusing times $R^*$. This lower bound is a function of the fidelity and unitarity~\cite{unitarity} of the noise channel $\Lambda$. For brevity, we omit the result here.


\subsection{Numerical simulations}
Let us first consider a practical experimental condition in which we possess the relaxation and decoherence times of the quantum gates~\cite{nielsen2010quantum}. These two types of times correspond to amplitude-damping and phase-damping noises, respectively. These two noise types are crucial for characterizing noise encountered in the physical realization of quantum systems.
The amplitude-damping channel $\Lambda_a$ describes an excited multi-level atomic state's decay due to the spontaneous emission of a photon, and the phase-damping channel $\Lambda_p$, also known as the dephasing channel, illustrates the decay of coherence over time due to external perturbations \cite{nielsen2010quantum}. The single-qubit amplitude-damping channel is 
\begin{equation}
\begin{split}
\Lambda_a(\rho)=K_0\rho K_0^{\dagger}+K_1\rho K_1^{\dagger},
\ K_0=\begin{pmatrix}
1&0\\
0&\sqrt{p}
\end{pmatrix},\ K_1=\begin{pmatrix}
0& \sqrt{1-p}\\
0&0
\end{pmatrix},
\end{split}
\end{equation}
and the single-qubit phase-damping channel is 
\begin{equation}
\begin{split}
\Lambda_p(\rho)=E_0\rho E_0^\dagger+E_1\rho E_1^\dagger,\ E_0=\begin{pmatrix}
1&0\\
0&\sqrt{p}
\end{pmatrix},
E_1=\begin{pmatrix}
0&0\\
0&\sqrt{1-p}
\end{pmatrix},
\end{split}
\end{equation}
where parameter $p\in [0,1]$ is related to the noise strength. Deriving an analytical expression for parameter $B$ in the single-qubit case is theoretically feasible but limited to simple noise models and not easily generalizable. It is challenging to extend it to multi-qubit systems, so we used numerical simulations to estimate $B$ and validate our results.

In our simulation, we take a composite noise channel comprising a local amplitude-damping channel $\Lambda_{la}$ with parameter $p_1$ alongside a local phase-damping channel $\Lambda_{lp}$ with parameter $p_2$,
\begin{equation}
\begin{split}\label{comp2}
\Lambda_{comp}=\Lambda_{lp}(p_2)\circ\Lambda_{la}(p_1),
\end{split}
\end{equation} 
where \begin{equation}
\begin{split}
\Lambda_{la}&=\Lambda_a^{\otimes n},\\
\Lambda_{lp}&=\Lambda_p^{\otimes n}.
\end{split}
\end{equation}

Consider the case that implementing each circuit $R$ times cost $t(R) = R + 4$. In this case, $R_0=4$ stands as a $2$-optimal value and the optimal reusing times $R^*=\sqrt{\frac{4(A-B)}{B-A^2}}$ according to Corollary \ref{cor:2opt_constant} and \ref{lemma1}, respectively.
Recall that once the noise channel form is determined, the value of $A$ can be evaluated directly via the noise channel fidelity, and the value of $B$ is a quantity averaged over multiple random Clifford gate sequences. We uniformly sample $50000$ distinct Clifford circuits at random for each predefined length $m$ to simulate the value of $B$. Subsequently, we compare the variance $\mathbb{V}(\Bar{X}_N)$ across different circuit reuse strategies: the optimal one $R^*$, the $2$-optimal one $R_0=4$ and the naive one $R_{naive}=1$.

The value of $R^*$ under noise channel $\Lambda_{comp}$ across various noise parameters is illustrated in Fig.~\ref{R_p1fix} and Fig.~\ref{R_p2fix}. One can observe that generally, when noise level increases, i.e., noise parameter $p_i$ becomes smaller, the value of $R^*$ tends to reduce. 
This is because when the noise channel incurs a larger deviation, more randomized circuit twirling is necessary to gather adequate information concerning the noise channel, consequently leading to less circuit reusing.
We also conduct a comparison between the variance of estimators for optimal $R^*$, $2$-optimal $R_0$, and that for naive $R_{naive}=1$, as illustrated in Fig.~\ref{V_p1fix} and Fig.~\ref{V_p2fix}. It can be seen that the variance for $R_0$ does not overpass the minimum variance $2$ times, consistent with our analysis of Corollary~\ref{cor:2opt_constant}. Moreover, both the optimal and $2$-optimal strategies demonstrate great improvement relative to the naive selection of $R_{naive}=1$. We also conduct simulations over other typical noise channels and analyze their performance, which can be found in Appendix~\ref{appendsc:Num_res}. 

\begin{figure}[htbp]
\raggedright
\vspace{-0.4cm}
\begin{minipage}[b]{0.9\linewidth}
\subfloat[]{\label{R_p1fix}
\includegraphics[width=7.8cm]{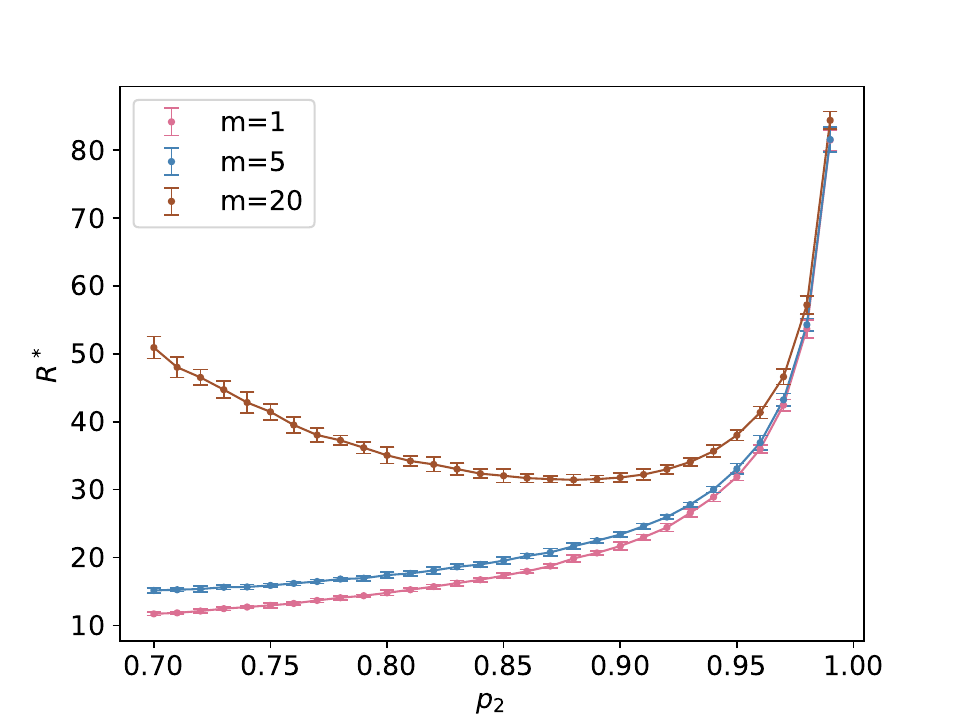}
}
\subfloat[]{\label{R_p2fix}
\includegraphics[width=7.8cm]{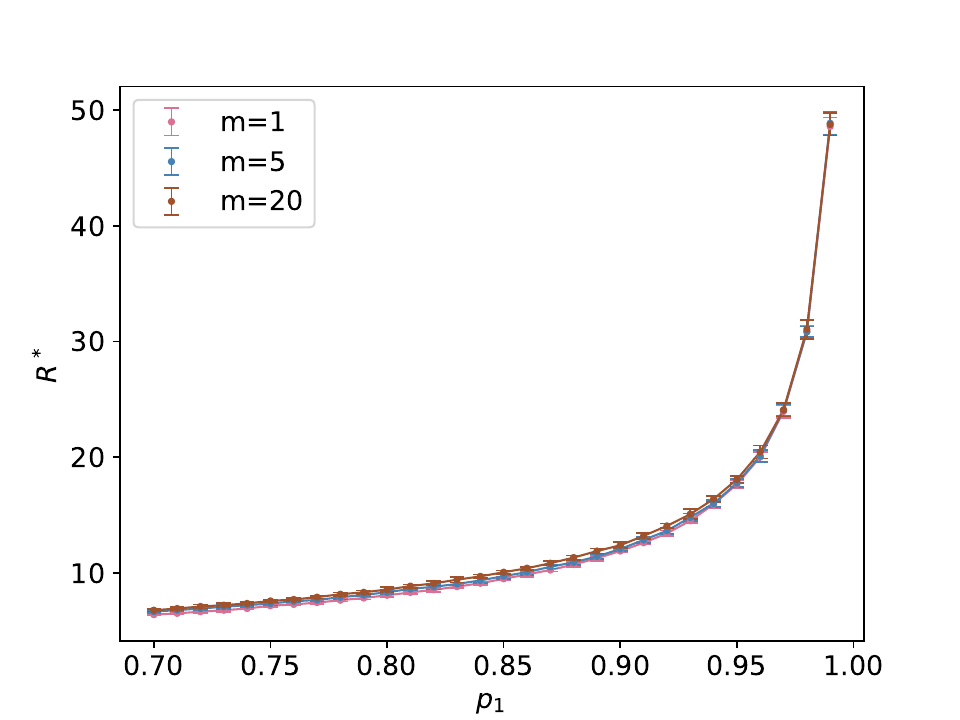}
}\\
\vspace{-0.4cm}
\end{minipage}
\begin{minipage}[b]{0.9\linewidth}
\subfloat[]{\label{V_p1fix}
\includegraphics[width=7.8cm]{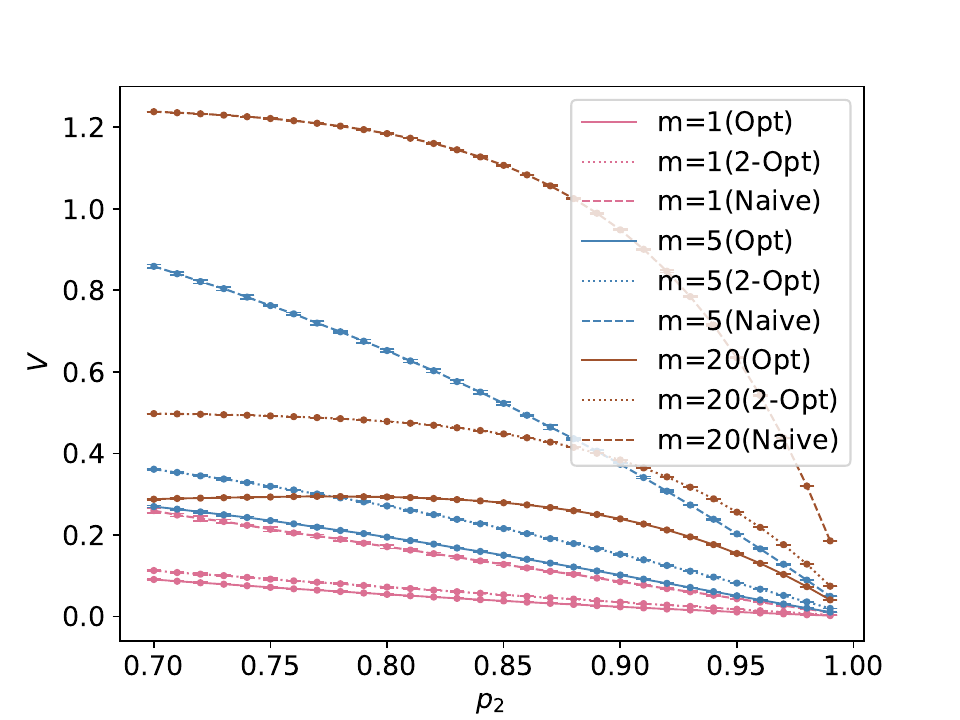}
}
\subfloat[]{\label{V_p2fix}
\includegraphics[width=7.8cm]{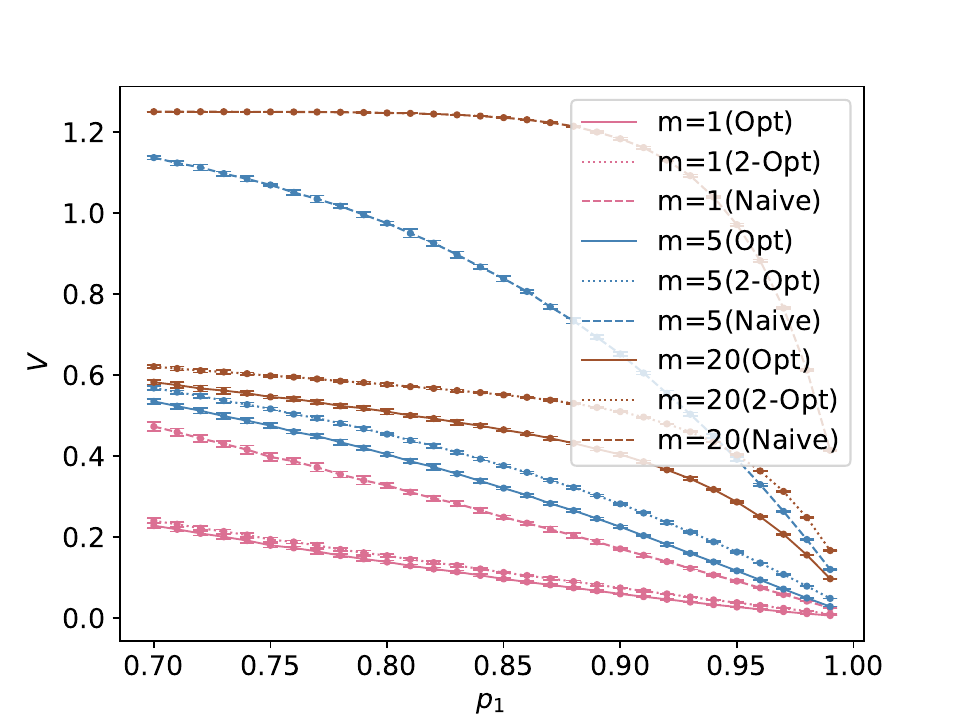}
}
\end{minipage}
\caption{These figures illustrate the simulation results for the standard RB protocol under the composite channel $\Lambda_{comp}$ defined by Eq.~\eqref{comp2}. The SPAM specifications are assigned as projective measurement $Q=\ketbra{0}$ and initial state $\rho=\ketbra{0}$. Figure (a) presents the optimal reusing times $R^*$ under the composite channel with fixed amplitude-damping parameter $p_1=0.999$, while Figure (b) shows $R^*$ with fixed phase-damping parameter $p_2=0.99$. Figure (c) compares the variance across the naive reusing parameter $R_{naive}=1$, the $2$-optimal one $R_0=4$ and the optimal one $R^*$, when the amplitude-damping parameter is fixed $p_1=0.999$, while figure (d) shows those with fixed phase-damping parameter $p_2=0.99$. We include error bars in the figures, representing twice the standard error of the mean.}
\label{Comp2_simulation}
\end{figure}

\section{Experimental results}\label{sc:Exp}
This section presents experimental results demonstrating the application of the circuit reusing strategy in RB. 
We conducted a $2$-qubit standard RB experiment on a superconducting quantum computing platform~\cite{Shaowei2022CZ}, employing various reusing times $R$. The set of reusing times in our experiment is presented below:
\begin{equation}
\mathbf{S} = \{2, 5, 10, 20, 50, 100, 200, 500, 1000, 2000, 5000, 10000, 20000\}.
\end{equation}
For each value of $R$, we implement $N$ distinct gate sequences, detailed in Table~\ref{runtime_table}.

\begin{table}[!htp]
\begin{adjustbox}{center}
\centering
\resizebox{1.1\linewidth}{!}{
\begin{tabular}{|c|c|c|c|c|c|c|c|c|c|c|c|c|c|}
\hline
$R$&2&5&10&20&50&100&200&500&1000&2000&5000&10000&20000\\
\hline
$N$&300000& 100000& 50000& 30000& 10000& 10000& 8000& 5000& 3000& 1800& 600& 300& 200\\
\hline
$T\ (s)$&52550.3& 17523.9& 8774.2& 5292.4& 1769.6& 1829.5& 1787.7& 1719.0& 1643.1& 1720.3& 1311.5& 1269.0& 1663.1\\
\hline
$T_0\ (s)$&53250.0& 17750.0& 8875.0& 5325.0& 1775.0& 1775.0& 1748.0& 1707.5& 1639.5& 1721.7& 1311.9& 1271.0& 1667.3\\
\hline
$T_0/T$&1.013& 1.013& 1.011& 1.006& 1.003& 0.970& 0.978& 0.993& 0.998& 1.001& 1.000& 1.002& 1.003\\
\hline
\end{tabular}
}
\end{adjustbox}
\caption{Comparison between the practical experimental runtime $T(N,R)$ and the estimated runtime $T_0(N,R)$ by our cost model $T_0(N,R)=C_1N\lceil R/R_c\rceil+C_2N$ with parameters $(C_1,C_2,R_c)=(0.0410s,0.1365s,100)$.}
\label{runtime_table}
\end{table}

The experimental runtime data, summarized in Table~\ref{runtime_table}, inspires us to adopt a ladder-like model for the total time expense $T_0$, dependent on $R$:
\begin{equation} \label{exp_cost_model}
T_0(N,R) = C_1N\lceil R/R_c \rceil + C_2N,
\end{equation}
where the coefficients are determined to be $(C_1, C_2, R_c) = (0.0410 \text{s}, 0.1365 \text{s}, 100)$. Accordingly, the cost to implement each circuit $R$ times is given by:
\begin{equation}
t(R) = C_1\lceil R/R_c \rceil + C_2,
\end{equation}
a model previously employed in the optimization analysis in Section \ref{subsec:R*}. To verify the accuracy of this cost model, we compare the experimental runtime per circuit $T/N$ against the values calculated using our model $T_0/N$, as shown in Fig.~\ref{fig:time_compare}.

\begin{figure}[htbp!]
\raggedright
\begin{minipage}[b]{0.9\linewidth}
\subfloat[]{\label{fig:time_compare}
\includegraphics[width=7.8cm]{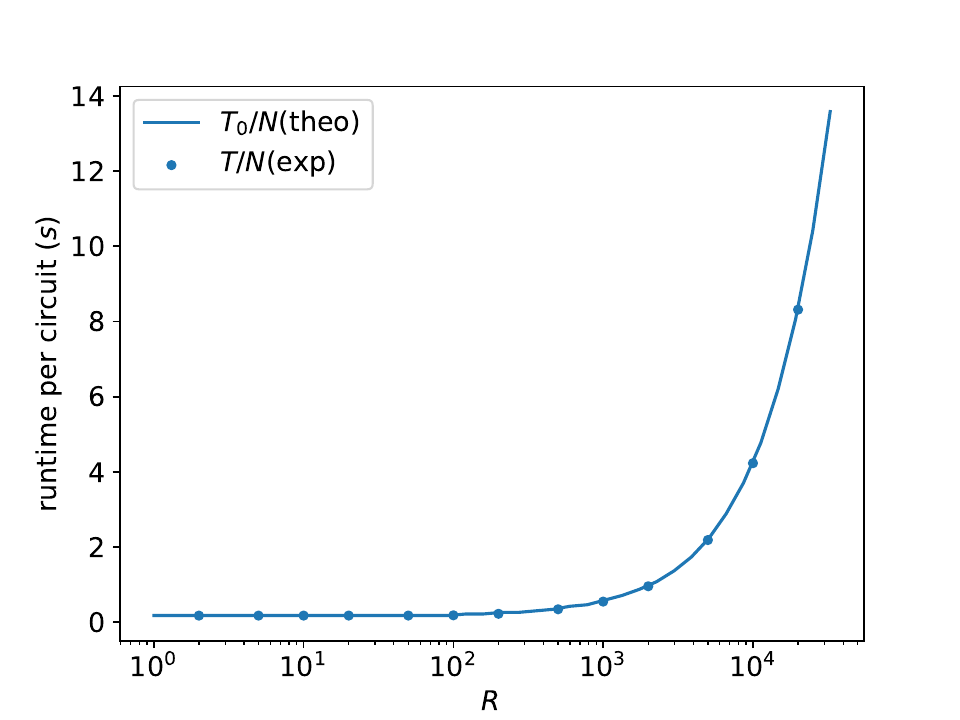}
}
\subfloat[]{
\label{fig:exp_V_R}
\includegraphics[width=7.8cm]{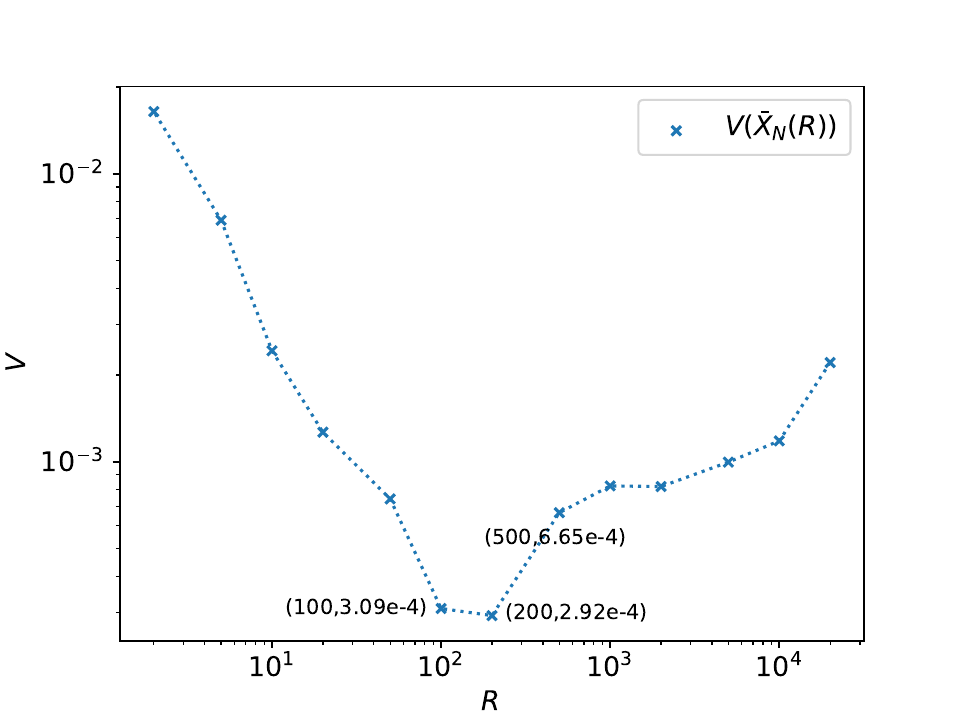}}
\end{minipage}
\caption{(a) This plot compares the experimental runtime per circuit, $T/N$, for different reusing times $R$, with the predicted costs from our model, $T_0/N$. The alignment of these values validates our cost model. (b) This plot displays the variance of the average measurement results, $\Bar{X}_N(R)$, across various reusing times $R$ and different numbers of unique gate sequences $N'$ as specified in Table~\ref{tab:assignedtask}. Each sequence is executed over a circuit length of $m=30$. The reusing times with the experimental minimum variance is 200, and the corresponding variance is $2.92\times 10^{-4}$. We also label the variances associated with the reusing times 100 and 500 in the figure.}
\end{figure}

From an experimental standpoint, the physical intuition behind this cost model presented in Eq.~\eqref{exp_cost_model} is clear. The term $C_2N$ represents the computational time expense required for generating each circuit, including the computation and decomposition of the inverse gate as well as the random selection of each gate sequence.
The term $C_1 N \lceil R/R_c \rceil$ represents the hardware execution time, which includes circuit implementation, data packing and transmission, and other associated tasks. Among these, the dominant time cost arises from data packing and transmission. In our quantum computing platform, data packing and transmission operations are performed in batches, with each batch containing $R_c$ units. After $R_c$ iterations of measurement, the results are transmitted to the classical computer as a single batch. This batching strategy is commonly used in classical computer architecture to amortize the computational overhead of specific operations across multiple instances, thereby improving system efficiency. Consequently, the factor $\lceil R/R_c \rceil$ indicates the number of ``batches'' into which the data transfer between quantum devices and classical computers is divided, while $C_1$ represents the cost per batch. In contrast, the time required for circuit implementation is relatively short and can be omitted.

For each reusing times $R\in\mathbf{S}$, we compute the variance of the average survival probability as depicted in Fig.~\ref{fig:exp_V_R}. To ensure uniform experimental costs across different values of $R$, we employ $N'$ distinct gate sequences for each value of $R$ to calculate the survival probability, where the values of $N'$ are detailed in Table~\ref{tab:assignedtask}. This process is repeated $10$ times to obtain 10 distinct survival probabilities for each $(R,N')$. Fig.~\ref{fig:exp_V_R} illustrates the variance among these 10 values. Specifically, the variances of the survival probabilities are $3.09 \times 10^{-4}$, $2.92\times 10^{-4}$, and $6.65 \times 10^{-4}$ for $R = 100$, $200$, and $500$, respectively. These values represent the three smallest variances observed in the experiments, suggesting that the experimental optimal reusing times $R^*_{exp}$ is approximately $200$, and lies within the range $[100,500]$.

\begin{table}[!htp]
\begin{adjustbox}{center}
\centering
\resizebox{1.1\linewidth}{!}{
\begin{tabular}{|c|c|c|c|c|c|c|c|c|c|c|c|c|c|}
\hline
$R$ & 2&5&10&20&50&100&200&500&1000&2000&5000&10000&20000 \\
\hline
$N'$ &715&715&715&715&715&715&581&372&232&133&58&30&15 \\
\hline
V&0.0164&0.0069&0.0024&0.0013&0.00074&0.00031&0.00029&0.00067&0.00082&0.00082&0.0010&0.0012&0.0022\\
\hline
$T_0$&125.05&127.09&126.82&127.21&126.79&127.04&126.95&126.91
&126.91&126.91&126.91&126.91&126.91\\
\hline
\end{tabular}
}
\end{adjustbox}
\caption{Experimental specifications for different reusing times $R$ with corresponding distinct gate sequences $N'$ to ensure uniform experimental runtime. Each set of measurements is repeated 10 times to evaluate the variance in survival probabilities shown in Fig.~\ref{fig:exp_V_R}.}
\label{tab:assignedtask}
\end{table}

Utilizing the condition that $R/R_c \leq \lceil R/R_c \rceil \leq (R+R_c-1)/(R_c)$, we derive that
\begin{equation}
\frac{C_1}{R_c}R + C_2 \leq t(R) \leq \frac{C_1}{R_c}R + C_2 + C_1\left(1-\frac{1}{R_c}\right).
\end{equation}
Based on Theorem~\ref{thm:2opt}, we determine that $R_0 = \frac{C_2R_c}{C_1}$ is a $\left(2 + \frac{C_1}{C_2}\left(1-\frac{1}{R_c}\right)\right)$-optimal value. Using the experimental coefficients $C_1, C_2$, and $R_c$, we find $R_0 = 333$ to be a $2.3$-optimal value. This theoretical finding is consistent with experimental data from $R = 200$ and $R = 500$, as illustrated in Fig.~\ref{fig:exp_V_R}.

Furthermore, we utilize the mean of the experimental measurement outcomes to approximate the expectations $A$ and $B$ as defined in Corollary~\ref{lemma1}, determining the theoretical optimal reusing times $R^*$.
\begin{equation}
\begin{split}
A &= \mathbb{E}_{\vec{G}}\left(\lbra{Q}\mathcal{G}_{inv}\Lambda\mathcal{G}_m\ldots\mathcal{G}_2\Lambda\mathcal{G}_1\lket{\rho}\right)\\
&= \mathbb{E}_{\vec{G}}\Pr(0|\vec{G})\\
&\approx \overline{\Pr(0)}\\
&= 0.1482, \\ 
B &= \mathbb{E}_{\vec{G}}\left(\lbra{Q}\mathcal{G}_{inv}\Lambda\mathcal{G}_m\ldots\mathcal{G}_2\Lambda\mathcal{G}_1\lket{\rho}^2\right)\\
&= \mathbb{E}_{\vec{G}}\Pr(0|\vec{G})^2\\
&\approx \overline{\Pr(0)^2}\\
&=0.0248, \\
Y &= A - B = 0.1234, \\ 
Z &= B - A^2 = 0.0028, \\ 
\sqrt{\frac{C_2Y}{C_1ZR_c}} &= 1.20, \\ 
R^* &\in \{100, 200\}.
\end{split}
\end{equation}
Here, $\Pr(0|\vec{G})$ denotes the probability of obtaining the measurement result $0$ when implementing the gate sequence $\vec{G}$; $\overline{\Pr(0)}$ represents the average probability of obtaining $0$ across all experimentally realized gate sequences; and $\overline{\Pr(0)^2}$ is the average of the squared probabilities of obtaining $0$ for all experimentally realized gate sequences. This result also aligns with the experimental optimal reusing times $R_{exp}^*$.

Note that despite the practical experiment for RB not strictly conforming to the gate-independent and Markovian noise assumption, the interval of our theoretical optimal reusing times remains in alignment with the experimental results and provides useful insights into reducing the experimental overhead.

\section{Conclusion and outlook}\label{sc:conclude}
Reusing quantum circuits often presents a more cost-effective approach compared to preparing and executing new circuits. In this work, we have developed a general analytical framework to optimize the circuit reusing times $R$ in quantum learning tasks. By employing the Law of Total Variance, we decompose the variance of outcomes into two components, enabling us to derive a near-optimal solution for determining $R$. This solution is remarkably versatile and applicable to various quantum learning tasks irrespective of specific noise models. Furthermore, when the noise characteristics of circuits are known, we can pinpoint the exact optimal reusing times $R^*$.

We experimentally implemented the standard RB protocol on a superconducting quantum computing platform to validate our theoretical framework. Utilizing the experimental runtime, we constructed a cost model specific to this platform and determined the optimal and near-optimal reusing times. The experimental data demonstrate that the initialization of a new circuit incurs a higher cost, underscoring the imperative of the circuit reusing strategy in experiments. Nonetheless, the relationship between $R$ and the total cost is observed to be non-linear, contradicting previous assumptions in the literature. Finally, the experimentally observed optimal reusing times, which minimizes experimental variance, aligns with our theoretical predictions.

For future exploration, our variance-cost analytical framework holds the potential to optimize protocols for characterizing various quantum properties, such as entanglement \cite{entanglement} and coherence \cite{streltsov2015measuring}. Our analysis thus far has assumed independent and identically distributed conditions for measurement outcomes. Future investigations could explore how our findings extend to scenarios where executed circuits exhibit correlations.  Additionally, it would be valuable to extend our analysis to a wider range of RB scenarios and investigate the reusing strategy within generalized RB frameworks~\cite{general,heinrich2022randomized}. Lastly, our results suggest revisiting similar studies in other quantum learning tasks, like thrifty shadow estimation \cite{Thrifty_shadow, zhou2023performance}, where the linear model might not be applicable.

\begin{acknowledgements}
We express our gratitude to Zhenyu Du, Chu Zhao, Zitai Xu, and Ziyi Xie for their insightful discussions. The experiments in this work are performed on the superconducting platform~\cite{Shaowei2022CZ} from the University of Science and Technology of China. We especially thank Daojin Fan and Ming Gong for the experimental implementation. This work was supported by the National Natural Science Foundation of China Grant No.~12174216 and the Innovation Program for Quantum Science and Technology Grant No.~2021ZD0300804 and No.~2021ZD0300702.
\end{acknowledgements}

\appendix

\section{Preliminary}\label{app_Rep}
\subsection{Representation theory}\label{rep_theory}
In this part, we recall some useful facts regarding the representations of finite groups. Consider a finite group, $\mathbf{G}$, and a finite-dimensional complex vector space $V$. Denote
$GL(V)$ as the group of invertible linear transformations of $V$, we can define a representation $\phi$ of the group $\mathbf{G}$ on the
space $V$ as a map \begin{equation}
\phi: \mathbf{G}\rightarrow GL(V):g \mapsto \phi(g),
\end{equation}
such that \begin{equation}
\phi(g)\phi(h)=\phi(gh),\ \  \forall g,h\in\mathbf{G}.
\end{equation}
Given representation $\phi$ on $V$, a linear subspace $W\subseteq V$ is called $invariant$ if \begin{equation}
\phi(g)w\in W,\ \ \forall w\in W,\forall g\in \mathbf{G}.
\end{equation}
The restriction of $\phi$ to the invariant subspace $W$ is known as a $subrepresentation$ of $\mathbf{G}$ on $W$. If there is a non-trivial subspace $W$ of $V$, i.e. $W\neq\{0\},V$, such that
\begin{equation}\label{irred}
\phi(g)W\subset W,\ \ \forall g\in \mathbf{G},
\end{equation}
then the representation $\phi$ is called $reducible$. If there are no non-trivial subspaces $W$ satisfying Eq.~\eqref{irred}, the representation 
$\phi$ is called $irreducible$. Two representations $\phi,\phi'$ of a group $\mathbf{G}$ on spaces $V,V'$ are called $equivalent$ if there exists an 
invertible linear map $T:V\rightarrow V'$ such that \begin{equation}
T\circ\phi(g)=\phi'(g)\circ T,\ \ \forall g\in \mathbf{G}.
\end{equation}

Maschke's lemma states that every representation $\phi$ of a group can be uniquely written as a direct sum of irreducible representations \begin{equation}
\phi(g)\simeq \bigoplus_{\lambda\in R_G} \phi_\lambda(g)^{\oplus m_\lambda},\ \ \forall g\in\mathbf{G},
\end{equation}
where $R_G$ is the index set of the irreducible representations, and $m_\lambda$ represents the multiplicity of the equivalent 
irreducible representation of $\phi_\lambda$ in $\phi$. 

One important object in representation theory is $character$, the $character$ $\chi_{\phi}$ corresponding to a representation $\phi$ of a group $\mathbf{G}$ is defined as \begin{equation}
\chi_{\phi}(g)=\Tr(\phi(g)).
\end{equation}
With the character, we introduce the $projector$ onto irreducible representation. Let $\phi_\lambda$ be an irreducible representation contained in $\phi$ with character $\chi_\lambda$, then the projector onto the support space of $\phi_\lambda$ can be 
described by \begin{equation}
\mathcal{P}_{\phi_\lambda}=\frac{|\phi_\lambda|}{|\mathbf{G}|}\sum_{g\in\mathbf{G}}\chi_\lambda(g)\phi(g),
\end{equation}
where $|\phi_\lambda|$ denotes the dimension of $\phi_\lambda$.

Define the twirl $\mathcal{T}_{\phi}(M)$ of a linear map $M:V\rightarrow V$ with respect to the representation $\phi$ to be \begin{equation}
\mathcal{T}_{\phi}(M):=\frac{1}{|\mathbf{G}|}\sum_{g\in\mathbf{G}}\phi(g)M\phi(g)^\dagger.
\end{equation}
Now we can further analyze the twirl map over specific types of representations by corollary of Schur's lemma: \begin{lemma}
Let $\mathbf{G}$ be a finite group and let $\phi$ be a representation of $\mathbf{G}$ on a complex vector space $V$ with decomposition \begin{equation}
\phi(g)\simeq \bigoplus_\lambda \phi_{\lambda}(g),\ \ \forall g\in\mathbf{G}
\end{equation}
into inequivalent irreducible subrepresentations $\phi_\lambda$, i.e. $m_\lambda=1$. Then for any liner operation $M:V\rightarrow V$, the twirl map of $M$ over $\mathbf{G}$ is given by \begin{equation}
\mathcal{T}_{\phi}(M)=\sum_\lambda\frac{\Tr(M\mathcal{P}_\lambda)}{\Tr(\mathcal{P}_\lambda)}\mathcal{P}_\lambda,
\end{equation}
where $\mathcal{P}_\lambda$ is the projector onto the subspace of $\phi_\lambda$.
\end{lemma}
Denote $f_\lambda=\frac{\Tr(M\mathcal{P}_\lambda)}{\Tr(\mathcal{P}_\lambda)}$, we can find that,
\begin{equation}
\mathcal{T}_\phi(M)^m=\sum_\lambda f_\lambda^m\mathcal{P}_\lambda,
\end{equation}
on account of the idempotent and mutually orthogonal properties of projectors $\mathcal{P}_\lambda$. 
Hence, to benchmark the gate set of a group $\mathbf{G}$, one can sample the circuit across the group with a given length $m$, then measure the averaged survival probabilities \begin{equation}
\mathbb{E}(X(m))=\sum_\lambda C_\lambda f_\lambda^m
\end{equation}
with coefficient $C_\lambda$ associated with SPAM terms. By applying post-processing to fit an exponential decay model, the average fidelity of the gate set can be estimated by the decay parameters $f_\lambda$, which is the key technique of RB \cite{general}. However, this functional form holds only when representation $\phi$ has no 
irreducible subrepresentation occurs more than once, that is, in the $m_\lambda=1$ case. In the general case with multiplicity $m_\lambda\neq 1$, the 
averaged survival probability takes the form \begin{equation}\label{non_inequ}
X(m)\approx \sum_\lambda \Tr(C_\lambda M_\lambda^m),
\end{equation}
where $M_\lambda$ is an $m_\lambda \times m_\lambda$ real matrix that only depends on the implementation of $\phi$ and $C_\lambda$ is an 
$m_\lambda\times m_\lambda$ matrix also associated with SPAM behavior \cite{general}. It is important to note that the matrices $M_\lambda$ may not be diagonalizable, rendering it impossible to determine the decay parameter associated with average fidelity by fitting in this scenario. Consequently, the majority of RB protocols concentrate primarily on scenarios where $m_\lambda=1$.

\subsection{Liouville representation of quantum channel}
In this part, we will introduce the Liouville representation, which is extensively utilized in discussions about quantum channels. 
Denote the Hilbert space for $n$ qubits as $\mathcal{H}$ and the set of linear operators on $\mathcal{H}$ 
as $\mathcal{L}(\mathcal{H})$. The Liouville representation is 
defined on a set of the trace-orthonormal basis on $\mathcal{L}(\mathcal{H})$, and the most widely used basis is the normalized Pauli group \begin{equation}
\mathbf{P}_n':=\left\{\sigma_i=\frac{P_i}{\sqrt{d}}\mid P_i\in\{\mathbb{I},X,Y,Z\}^{\otimes n}\right\},
\end{equation}
where $d=2^n$ is the system dimension. A good property of this basis is that each normalized Pauli operator is complete and satisfies \begin{equation}
\Tr\left(\sigma_i^\dagger\sigma_j\right)=\delta_{ij},\ \ \forall\sigma_i,\sigma_j\in\mathbf{P}_n'
\end{equation}
under the Hilbert-Schmidt inner product. Any $n$-qubit operator $O\in\mathcal{L}(\mathcal{H})$ can be decomposed over the $4^n$ normalized Pauli operators, \begin{equation}
O=\sum_{\sigma_i\in\mathbf{P}_n'}\Tr\left(\sigma_i^\dagger O\right)\sigma_i.
\end{equation} 
Thus, one can reformulate $O$ in a vector form, \begin{equation}
\lket{O}=\sum_{\sigma_i\in\mathbf{P}_n'}\Tr\left(\sigma_i^\dagger O\right)\lket{\sigma_i},
\end{equation}
where $\lket{\sigma_i}$ in Liouville representation is a $4^n$-length vector with only one non-zero element $1$. 

Recall that quantum channels are known as completely positive and 
trace-preserving linear maps on $\mathcal{L}(\mathcal{H})$. Through the Liouville representation, any 
quantum channel can be represented as a matrix. The action of the channel $\Lambda$ on an operator $O$ is given by \begin{equation}
\lket{\Lambda(O)}=\Lambda\lket{O},
\end{equation}
where $\Lambda_{ij}=\lbra{\sigma_i}\Lambda\lket{\sigma_j}=\Tr(\sigma_i\Lambda(\sigma_j))$. In the main text, what we mainly use is the action of a unitary operation $G$ on a state $\rho$: \begin{equation}
\lket{G\rho G^\dagger}=\mathcal{G}\lket{\rho}.
\end{equation} 
The composition of two channels can be depicted as the product of two matrices, \begin{equation}
\lket{\Lambda_2\circ\Lambda_1(O)}=\Lambda_2\Lambda_1\lket{O}.
\end{equation}
Furthermore, the measurement operator can also be vectorized with the Liouville notation according to the definition of the Hilbert-Schmidt inner product. For example, the measurement probability of a state $\rho$ on a positive operator-valued measure (POVM) $\{F_i\}$ is described as \begin{equation}
p_i=\lbraket{F_i}{\rho}=\Tr\left(F_i^\dagger\rho\right).
\end{equation}

\subsection{Randomized benchmarking}\label{appendssc:rb}
The following box presents the detailed procedure of standard RB with $R$ times circuit reusing. 
\begin{tcolorbox}[title = {Standard RB}, breakable]
\begin{itemize}
\item[1.] Sample a random circuit $\vec{G}_i=\{G_1, \ldots, G_m\}$ of length $m$, where each gate $G_i$ is independently 
and uniformly sampled from the gate set $\mathbf{G}$. 
\item[2.] Compute the global inverse gate $G_{inv} = (G_m\cdots G_1)^{-1}$, and add it to the end of $\vec{G}_i$ to ensure that the entire circuit is equivalent to the identity channel $\mathcal{I}$ in the noiseless case. 
\item[3.] Apply the circuit with inverse gate to the initial state $\rho$, and then take POVM measurement $\{Q,\mathbb{I}-Q\}$. 
\item[4.] Repeat step 3 for $R$ times, the result of $r$-th measurement can be described as a random variable $X_r(m,\vec{G_i})$, and estimate the survival probability of sequence $\vec{G_i}$: \begin{equation}
\Bar{X}_R(m,\vec{G}_i)=\frac{1}{R}\sum_rX_r(m,\vec{G_i}).
\end{equation} 
\item[5.] Repeat steps 1-4 for $N$ distinct circuits $\vec{G_i}$, and estimate the overall survival probability by aggregating the result derived from each set of measurements: \begin{equation}
\Bar{X}_N(m)=\frac{1}{N}\sum_i\Bar{X}_R(m,\vec{G_i}).
\end{equation}
\item[6.] Repeat steps 1-5 across a sufficient range of different lengths $m$ and fit $\Bar{X}_N(m)$ to the following decay function to get $f(\Lambda)$: \begin{equation}
\Bar{X}_N(m)=af(\Lambda)^m+b.
\end{equation}
\end{itemize}
\end{tcolorbox}
In summary, by random sampling, we obtain the survival probabilities of circuits for different lengths $m$ and fit them to an exponential decay function to get the average fidelity of the gate set $\mathbf{G}$ \begin{equation}
F_{avg}(\Lambda)=f(\Lambda)+\frac{1-f(\Lambda)}{d}.    
\end{equation} 
Note that the initial state $\rho$ and POVM measurement $\{Q,\mathbb{I}-Q\}$ are fixed.


\section{Optimizing circuit reusing times with circuit and noise information for generic cost models}\label{appendsc:opt}
In this part, we discuss the optimal reusing times for generic cost models. Recall that for a quantum learning task, the variance of the estimator is described as
\begin{equation}\label{eq:varianceappend}
\begin{split}
\mathbb{V}(\Bar{X}_N(R))=\frac{t(R)}{T_0}\left(\frac{Y}{R}+Z\right),
\end{split}
\end{equation}
where $Y$ and $Z$ are coefficients associated with measurement results.
Normally, the optimal reusing times can be determined by minimizing Eq.~\eqref{eq:varianceappend} once the specific form of cost model $t(R)$ is given. Here, we discuss a more general case where we only possess prior knowledge about the bound of the cost $t(R)$:
\begin{equation}
\beta_l R+\alpha_l \leq t(R)\leq \beta_u R+\alpha_u,\ \forall R\in\mathbb{Z}^+
\end{equation}
with constants $\alpha_l,\alpha_u,\beta_l,\beta_u$ all positive. This is the case that the real cost model $t(R)$ is bounded by two constant-cost models, as mentioned in the main text. Given the values of $Y$ and $Z$, we can identify the range of the optimal solution to minimize the variance, as shown in the following theorem.

\begin{theorem}\label{thm:boundopt}
Given the bound of the cost model $t(R)$, $\beta_l R+\alpha_l \leq t(R)\leq \beta_u R+\alpha_u$ with $\alpha_l,\alpha_u,\beta_l,\beta_u$ all positive, within fixed total time cost, the optimal reusing times $R^*$ to minimize the variance $\mathbb{V}(\Bar{X}_N(R))$ satisfies
\begin{equation}
R^*\in [\frac{b-\sqrt{b^2-4ac}}{2a}, \frac{b+\sqrt{b^2-4ac}}{2a}],
\end{equation}
with
\begin{equation}
\begin{split}
a &= \beta_l Z;\\
b &= (\sqrt{\alpha_u Z} + \sqrt{\beta_u Y})^2 - (\alpha_l Z + \beta_l Y);\\
c &= \alpha_l Y.
\end{split}
\end{equation}
\end{theorem}
\begin{proof}
Since $t(R)\leq \beta_u R+\alpha_u$, we have
\begin{equation}
\mathbb{V}(\Bar{X}_N(R)) \leq \frac{\beta_u R+\alpha_u}{T_0}\left(\frac{Y}{R}+Z\right).
\end{equation}
Thus,
\begin{equation}
\begin{split}
\mathbb{V}(\Bar{X}_N(R^*)) &= \min_R \mathbb{V}(\Bar{X}_N(R))\\
&\leq \min_R\frac{\beta_u R+\alpha_u}{T_0}\left(\frac{Y}{R}+Z\right)\\
&= \frac{(\sqrt{\alpha_u Z} + \sqrt{\beta_u Y})^2}{T_0}.
\end{split}
\end{equation}
On the other hand, since $t(R)\geq \beta_l R+\alpha_l$, we have
\begin{equation}
\mathbb{V}(\Bar{X}_N(R)) \geq \frac{\beta_l R+\alpha_l}{T_0}\left(\frac{Y}{R}+Z\right).
\end{equation}
By substituting $R^*$ into the above inequality, we get
\begin{equation}
\frac{\beta_l R^*+\alpha_l}{T_0}\left(\frac{Y}{R^*}+Z\right) \leq \frac{(\sqrt{\alpha_u Z} +\sqrt{\beta_u Y})^2}{T_0}.
\end{equation}
The above inequality is equivalent to $R^*\in [\frac{b-\sqrt{b^2-4ac}}{2c}, \frac{b+\sqrt{b^2-4ac}}{2c}]$, which finishes the proof.
\end{proof}

Theorem~\ref{thm:boundopt} provides a bound for the optimal reusing times once the cost model is bounded. Note that when the cost model satisfies $t(R) = R + \alpha$ with a constant $\alpha$, the result of Theorem~\ref{thm:boundopt} becomes 
$R^*=\sqrt{\frac{\alpha Y}{Z}}$, which is consistent with Theorem~\ref{thm:optimalR}. In this sense, Theorem~\ref{thm:boundopt} can be viewed as a generalization of Theorem~\ref{thm:optimalR}.

\section{Analysis of optimal reusing times under various noise channels}\label{appendsc:Num_res}
In experiments, we can utilize prior knowledge about the noise channel to estimate optimal reusing values.
Here, we analyze the optimal reusing times $R^*$ for standard RB by numerical simulation across typical noise channels, specifically: 
\begin{equation}
\begin{split}
&\text{global depolarizing channel: }
\Lambda_{dep}(\rho)=p\rho+(1-p)\frac{\mathbb{I}}{2^n},\\
&\text{local amplitude-damping channel: } \Lambda_{la}=\Lambda_a^{\otimes n},\\
&\text{where }\Lambda_a(\rho)=K_0\rho K_0^{\dagger}+K_1\rho K_1^{\dagger},
K_0=\begin{pmatrix}
1&0\\
0&\sqrt{p}
\end{pmatrix},\ K_1=\begin{pmatrix}
0& \sqrt{1-p}\\
0&0
\end{pmatrix},\\
&\text{local phase-damping noise channel: }\Lambda_{lp}=\Lambda_p^{\otimes n},\\
&\text{where }\Lambda_p(\rho)=E_0\rho E_0^\dagger+E_1\rho E_1^\dagger,E_0=\begin{pmatrix}
1&0\\
0&\sqrt{p}
\end{pmatrix},
E_1=\begin{pmatrix}
0&0\\
0&\sqrt{1-p}
\end{pmatrix},\\
&\text{local unitary noise channel: } \Lambda_{lz}=\Lambda_z^{\otimes n},\text{where }\Lambda_z(\rho)=\exp(i\theta \mathcal{Z})(\rho)=i\sin\theta \mathcal{Z}(\rho)+\cos\theta\cdot\rho,
\end{split}
\end{equation}
where $\mathcal{Z}(\rho)=Z\rho Z^\dagger$, $n$ denotes the system qubit number, and $p\in[0,1],\theta\in[0,\pi]$ are related parameters of the channels. 
To cover multi-qubit noise channels, we also consider a ``correlation channel'': \begin{equation}
\Lambda_c = \bigotimes_{i<j}\exp\left(i\beta_{ij}\text{SWAP}_{ij}\right),
\end{equation}
where $\beta_{ij}\in[0,\pi]$ is the correlation parameter representing the interaction strength between qubits $i$ and $j$. These models can effectively approximate the vast majority of noises encountered in practical experiments. 

To depict the trend of $R^*$ under different kinds of noise channel parameters, we consider the case that reusing an existing circuit $R$ times cost $t(R) = R + 4$. In this case, $R_0=4$ stands as a $2$-optimal value according to Corollary \ref{cor:2opt_constant}, and the optimal reusing times reads $R^*=\sqrt{\frac{A-B}{B-A^2}}$. If the actual cost factor is not equal to $4$ but another $\alpha^*$, the optimal reusing times obtained from simulations can be adjusted by multiplying the factor $\sqrt{\alpha^*}$, as the theoretical optimal $R^*$ is proportional to $\sqrt{\alpha}$ for the constant-cost model $t(R) = R+\alpha$.
To determine $R^*$, one has to get the values of $A$ and $B$. The value of $A$ can be evaluated directly via noise channel fidelity. The value of $B$ is a quantity averaged over multiple random Clifford gate sequences. We uniformly sample $50000$ distinct Clifford circuits at random for each predefined length $m$ to simulate the value of $B$. Then, we can observe how much the optimal $R^*$ reusing strategy and the $2$-optimal $R_0$ reusing strategy improve compared to naively choosing $R_{naive}=1$. Within this section, for the $n$-qubit simulations the SPAM specifications are assigned as $\rho=\ketbra{0}^{\otimes n}$ and $Q=\ketbra{0}^{\otimes n}$. 



\subsection{Global depolarizing noise channel}
Depolarizing noise is an important environmental noise characterization for a quantum channel~\cite{nielsen2010quantum}.
Under the global depolarizing noise channel, abbreviated as ``depolarizing channel'' without ambiguity, we observe that the optimal reusing times $R^*\rightarrow\infty$ 
for the standard RB protocol. This occurs because the final state under application of arbitrary circuit $\vec{G}=\{G_1,G_2,\ldots,G_m,G_{inv}\}$, \begin{equation}
\begin{split}  \rho_f(\vec{G})&=\mathcal{G}_{inv}\Lambda_{dep}\mathcal{G}_m\Lambda_{dep}\ldots\Lambda_{dep}\mathcal{G}_1(\rho)\\
&=p\mathcal{G}_{inv}\Lambda_{dep}\mathcal{G}_m\Lambda_{dep}\ldots\Lambda_{dep}\mathcal{G}_2\mathcal{G}_1(\rho)+(1-p)\mathcal{G}_{inv}\Lambda_{dep}\mathcal{G}_m\Lambda_{dep}\ldots\Lambda_{dep}\mathcal{G}_2(\frac{\mathbb{I}}{2^n})\\
&=p^m\mathcal{G}_{inv}\mathcal{G}_m\ldots\mathcal{G}_1(\rho)+(1-p^m)(\frac{\mathbb{I}}{2^n})\\
&=p^m(\rho)+(1-p^m)(\frac{\mathbb{I}}{2^n}),
\end{split}
\end{equation}
is solely dependent on the noise channel parameter $p$ and circuit length $m$, where $\mathcal{G}_i$ represents the Liouville representation for $G_i$.
Consequently, the quality parameter can be extracted without the need for twirling over the entire gate set. Instead, one can simply reuse a single circuit $N$ times to estimate the noise channel and minimize the experimental cost.
The underlying reason for the result $R^*\rightarrow\infty$ is the commutative property of the global depolarizing channel with respect to quantum unitary operations. Hence, this result is not only applicable to the standard RB protocol but also to other RB protocols where the overall circuit becomes identity channel $\mathcal{I}$ for the noiseless case \cite{general}. 

However, experimental noise is typically not limited to depolarizing noise alone. To analyze the impact of depolarizing noise on the optimal reusing times $R^*$, we consider a composite noise channel comprising a local amplitude-damping channel with parameter $p_1$ alongside a depolarizing channel with parameter $p_2$: \begin{equation}\label{comp}
\Lambda_{comp2}(\rho)=\Lambda_{dep}(p_2)\circ\Lambda_{la}(p_1)(\rho),
\end{equation}
with fixed amplitude-damping parameter $p_1=0.99$. The simulation results for $R^*$ and comparison of variance across the optimal $R^*$, $2$-optimal $R_0$, and naive $R_{naive}$ are illustrated in Fig.~\ref{fig:V_comp1}. As $p_2$ approaches $0$, $R^*$ increases since the depolarizing part becomes more dominant, and it is known to us that $R^*\rightarrow\infty$ for the pure depolarizing channel. 
On the other hand, when $p_2$ is close to $1$, $R^*$ approaches the result under pure local amplitude-damping channel with given parameter $p_1$.

\begin{figure}[htbp]
\raggedright
\begin{minipage}[b]{0.9\linewidth}
\subfloat[]{
\includegraphics[width=7.8cm]{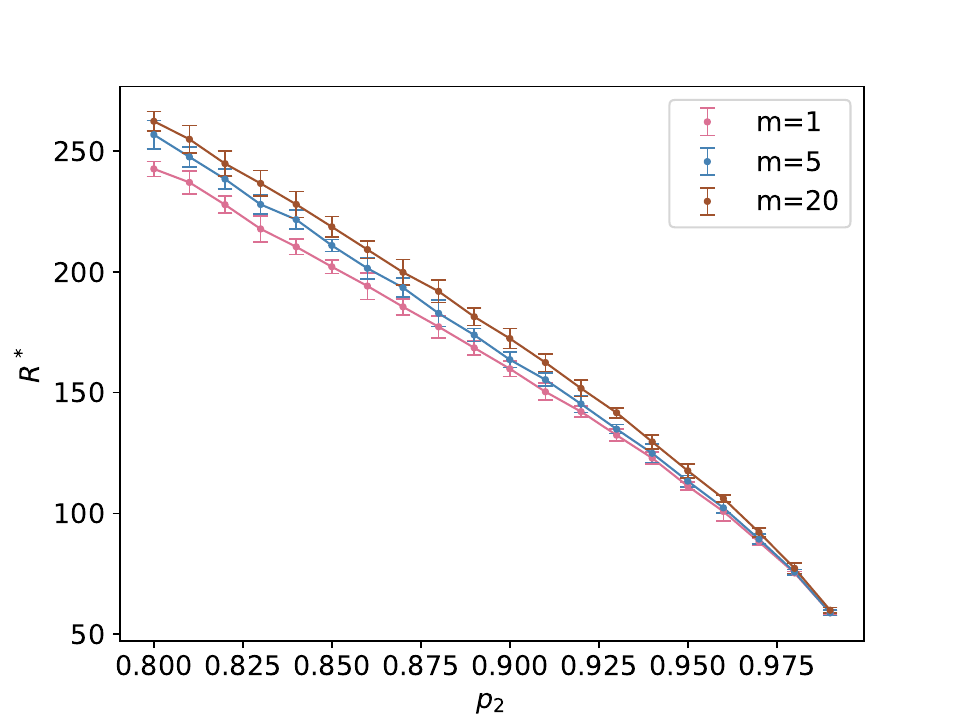}
}
\subfloat[]{
\includegraphics[width=7.8cm]{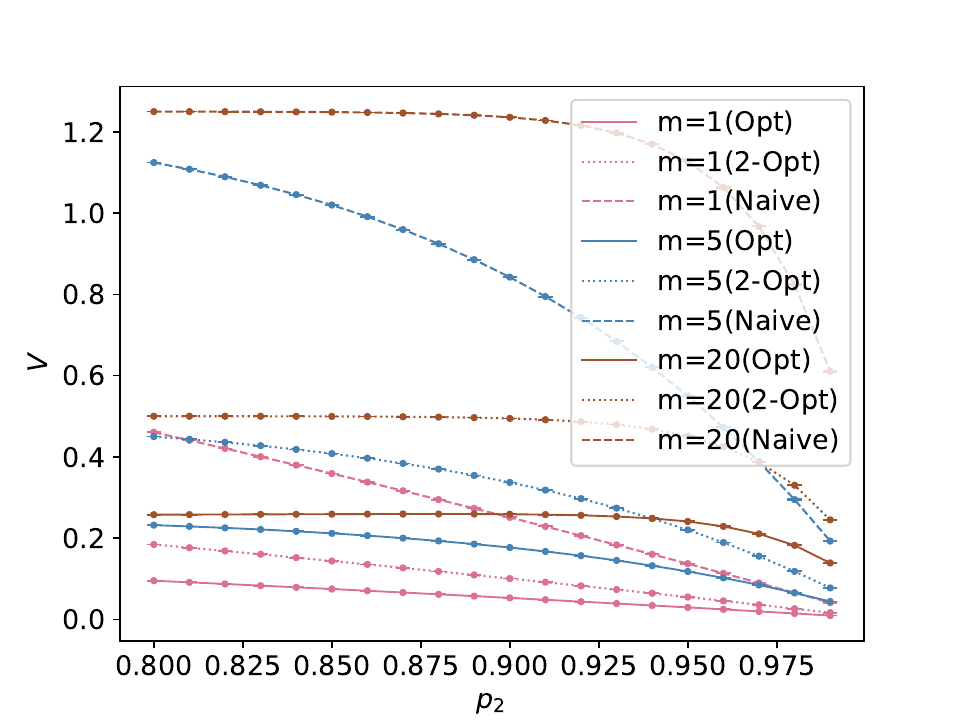}
}
\end{minipage}
\caption{Figure (a) presents the optimal reusing times $R^*$ for standard RB over composite noise channel $\Lambda_{comp2}$ defined by Eq.~\eqref{comp} with parameter $p_1=0.99$ and $p_2$ varying as the x-axis, while figure (b) shows a comparison between the variance of the estimator under the optimal $R^*$, under the $2$-optimal $R_0=4$, and under the naive $R_{naive}=1$ within same noise scenario.}
\label{fig:V_comp1}
\end{figure}

\subsection{Local unitary noise channel}
For the local unitary noise channel in standard RB, let us first consider an extreme case with the noise channel parameter $\theta=\frac{\pi}{2}$: 
$\Lambda_{lz}(\theta=\frac{\pi}{2})=(i\mathcal{Z})^{\otimes n}$. The final state under this channel can be described as \begin{equation}     
\begin{split}
\rho_f(\vec{G}, 
m)&=\mathcal{G}_{inv}\Lambda_{lz}\mathcal{G}_m\Lambda_{lz}\ldots\Lambda_{lz}\mathcal{G}_1(\rho)\\
&=(\bar{\mathcal{G}}_m^{\dagger}\mathcal{Z}^{\otimes n}\bar{\mathcal{G}}_m)(\bar{\mathcal{G}}_{m-1}^{\dagger}\mathcal{Z}^{\otimes n}\bar{\mathcal{G}}_{m-1})\ldots(\bar{\mathcal{G}}_1^{\dagger}\mathcal{Z}^{\otimes n}\bar{\mathcal{G}}_1)(\rho)\\
&=\mathcal{P}_m\mathcal{P}_{m-1}\ldots \mathcal{P}_1(\rho)\\
&=\mathcal{P}'(\rho),
\end{split} \end{equation}
where $\bar{\mathcal{G}}_i=\mathcal{G}_i\mathcal{G}_{i-1}\ldots\mathcal{G}_1$, and $\mathcal{P}_i(\rho)=P\rho P^\dagger$ for non-identity $n$-qubit Pauli operator $P\in \mathbf{P}_n\backslash \{\mathbb{I}\}$. This implies that each circuit $\vec{G}_i$ corresponds to a random Pauli operator $\mathcal{P}'$ under the special local unitary noise channel $\Lambda_{lz}(\theta=\frac{\pi}{2})$. As a result, the value of $\Tr(Q\rho_f(\vec{G}_i))$ is 
either $0$ or $1$, then \begin{equation}
\begin{split}
\Tr(Q\rho_f)&=\Tr(Q\rho_f)^2,\\
A=\mathbb{E}(\Tr(Q\rho_f))&=\mathbb{E}(\Tr(Q\rho_f)^2)=B,
\end{split}
\end{equation}
the reusing times $R^*$ has a value 
equal to $0$. 
On the other side, when $\theta\rightarrow0$, the noise channel is close to identity channel $\mathcal{I}$, $B\approx A^2\approx 1$, so we have $R^*\rightarrow\infty$.

\subsection{Two-qubit correlation noise channel}
As shown in Fig.~\ref{two_cor}, we investigate a $2$-qubit composite channel \begin{equation}\label{comp3}
\Lambda_{comp3}=\Lambda_c(\beta_{12})\circ\Lambda_{lp}(p_2)\circ\Lambda_{la}(p_1).
\end{equation}

\begin{figure}[htbp]
\raggedright
\begin{minipage}[b]{0.9\linewidth}
\subfloat[]{
\includegraphics[width=7.8cm]{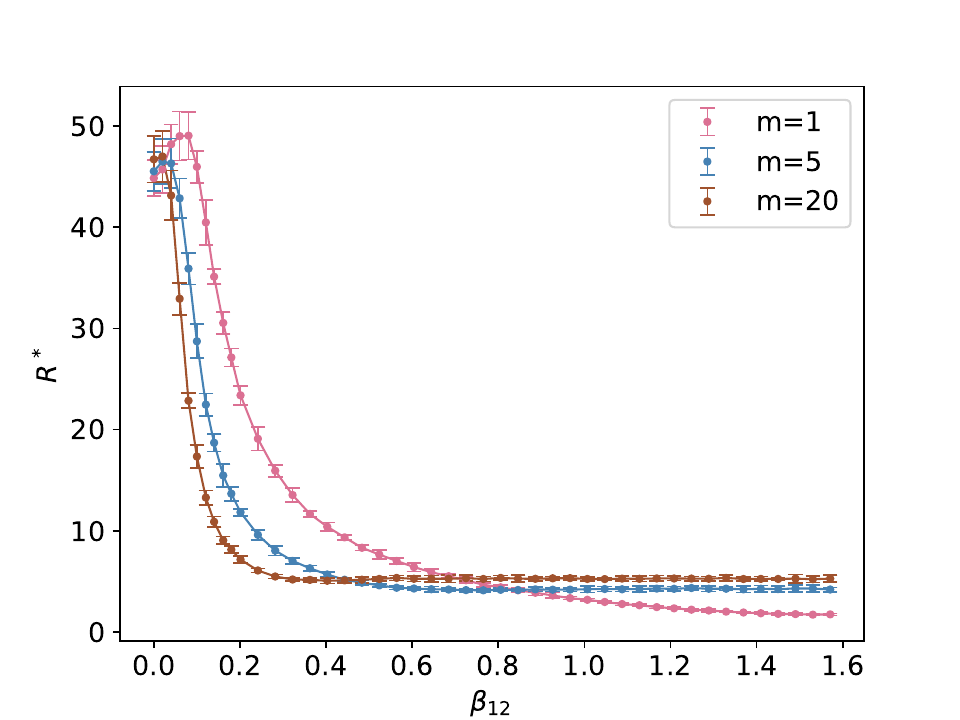}
}
\subfloat[]{
\includegraphics[width=7.8cm]{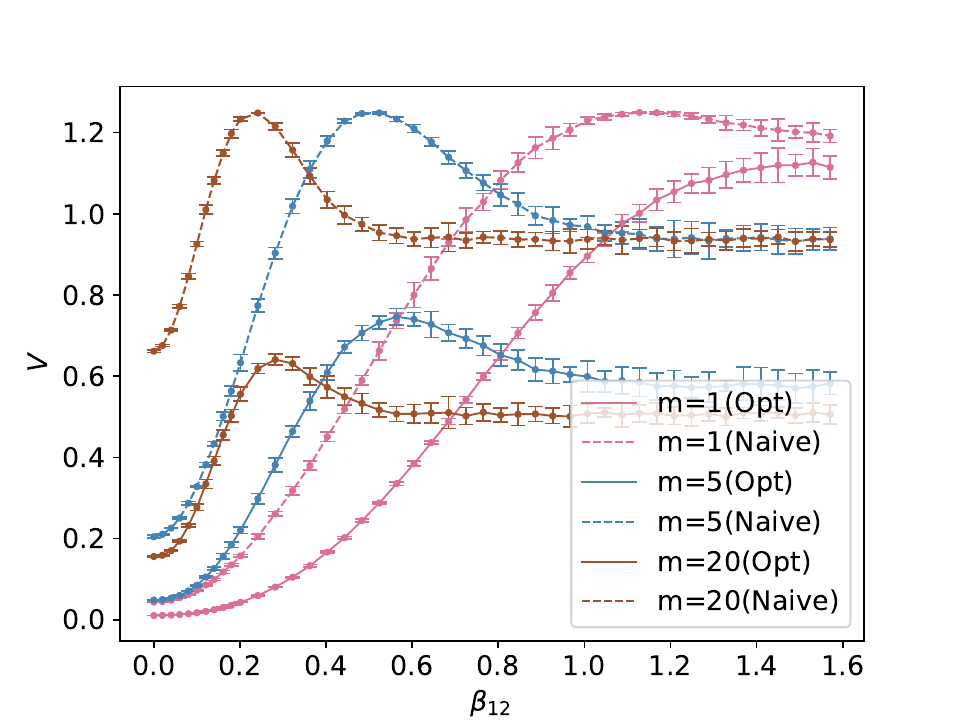}
}
\end{minipage}
\caption{Figure (a) presents the optimal reusing times $R^*$ for the $2$-qubit standard RB protocol over composite noise channel $\Lambda_{comp3}$ defined by Eq.~\eqref{comp3}. Figure (b) presents a comparison between the variance under the optimal $R^*$ and that under the naive $R_{naive}=1$. We omit the performance of the 2-optimal strategy $R_0=4$ in the figures due to its variance being nearly identical to that of $R^*$, which would otherwise congest the visualization. The results under the $\Lambda_{comp3}$ channel exhibit symmetry with respect to $\beta_{12}=\frac{\pi}{2}$; hence, only the range of $[0, \frac{\pi}{2}]$ for $\beta_{12}$ is depicted in the figure.}
\label{two_cor}
\end{figure}

The parameter $\beta_{12}$ controls the noise strength of the interaction between two subsystems. The amplitude-damping and phase-damping parameters are assigned as $p_1=0.999$ and $p_2=0.99$, respectively. When the noise parameter $\beta_{12}$ increases, the value of $R^*$ decreases. This occurs because as $\beta_{12}$ approaches $\frac{\pi}{2}$, the noise induced by $\Lambda_c(\beta_{12})$ intensifies, necessitating more twirling to learn about the noise channel. Note that in Fig.~\ref{two_cor}, $R^*$ first slightly increases and then decreases as $\beta_{12}$ approaching $\frac{\pi}{2}$. The increase results from the influence of another two types of composite noises $\Lambda_{la}$ and $\Lambda_{lp}$. In the case where the noise channel is exclusively the correlation channel $\Lambda_c$, $R^*$ will exhibit a monotonic decrease as $\beta_{12}$ increases, as illustrated in Fig.~\ref{fig:exclu_cor}.

\begin{figure}
    \raggedright
    \begin{minipage}[b]{0.9\linewidth}
    \subfloat[]{
    \includegraphics[width=7.8cm]{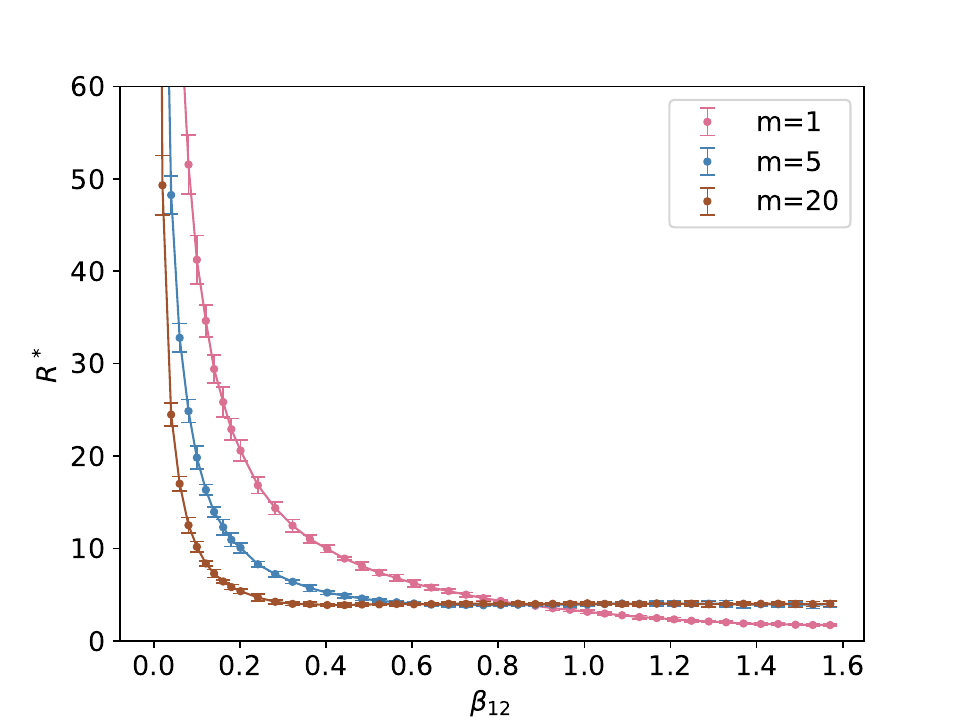}
    }
    \subfloat[]{
    \includegraphics[width=7.8cm]{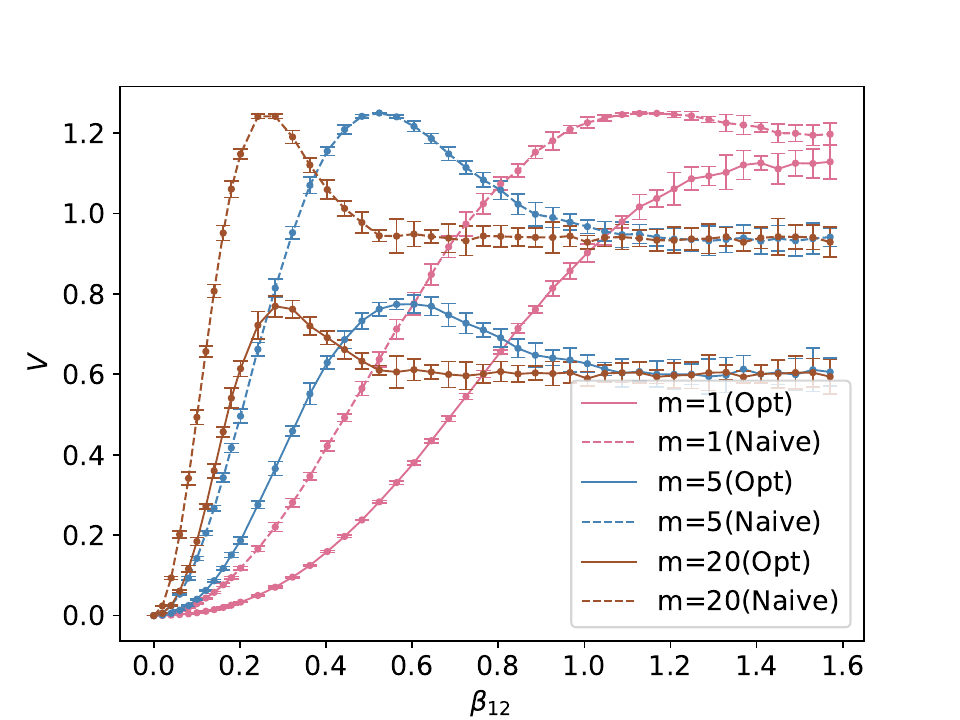}
    }
    \end{minipage}
    \caption{Figure (a) illustrates the values of the optimal reusing times $R^*$ under the correlation channel $\Lambda_c$. Figure (b) presents a comparison between the variance under the optimal $R^*$ and that under the naive $R_{naive}=1$.}
    \label{fig:exclu_cor}
\end{figure}
\bibliographystyle{apsrev}

\bibliography{bibThriftyRB.bib}

\begin{thebibliography}{36}
\expandafter\ifx\csname natexlab\endcsname\relax\def\natexlab#1{#1}\fi
\expandafter\ifx\csname bibnamefont\endcsname\relax
  \def\bibnamefont#1{#1}\fi
\expandafter\ifx\csname bibfnamefont\endcsname\relax
  \def\bibfnamefont#1{#1}\fi
\expandafter\ifx\csname citenamefont\endcsname\relax
  \def\citenamefont#1{#1}\fi
\expandafter\ifx\csname url\endcsname\relax
  \def\url#1{\texttt{#1}}\fi
\expandafter\ifx\csname urlprefix\endcsname\relax\def\urlprefix{URL }\fi
\providecommand{\bibinfo}[2]{#2}
\providecommand{\eprint}[2][]{\url{#2}}

\bibitem[{\citenamefont{Cross et~al.}(2015)\citenamefont{Cross, Smith, and Smolin}}]{cross2015quantum}
\bibinfo{author}{\bibfnamefont{A.~W.} \bibnamefont{Cross}}, \bibinfo{author}{\bibfnamefont{G.}~\bibnamefont{Smith}}, \bibnamefont{and} \bibinfo{author}{\bibfnamefont{J.~A.} \bibnamefont{Smolin}}, \bibinfo{journal}{Phys. Rev. A} \textbf{\bibinfo{volume}{92}}, \bibinfo{pages}{012327} (\bibinfo{year}{2015}), \urlprefix\url{https://doi.org/10.1103/PhysRevA.92.012327}.

\bibitem[{\citenamefont{Huang et~al.}(2022)\citenamefont{Huang, Broughton, Cotler, Chen, Li, Mohseni, Neven, Babbush, Kueng, Preskill et~al.}}]{huang2022quantum}
\bibinfo{author}{\bibfnamefont{H.-Y.} \bibnamefont{Huang}}, \bibinfo{author}{\bibfnamefont{M.}~\bibnamefont{Broughton}}, \bibinfo{author}{\bibfnamefont{J.}~\bibnamefont{Cotler}}, \bibinfo{author}{\bibfnamefont{S.}~\bibnamefont{Chen}}, \bibinfo{author}{\bibfnamefont{J.}~\bibnamefont{Li}}, \bibinfo{author}{\bibfnamefont{M.}~\bibnamefont{Mohseni}}, \bibinfo{author}{\bibfnamefont{H.}~\bibnamefont{Neven}}, \bibinfo{author}{\bibfnamefont{R.}~\bibnamefont{Babbush}}, \bibinfo{author}{\bibfnamefont{R.}~\bibnamefont{Kueng}}, \bibinfo{author}{\bibfnamefont{J.}~\bibnamefont{Preskill}}, \bibnamefont{et~al.}, \bibinfo{journal}{Science} \textbf{\bibinfo{volume}{376}}, \bibinfo{pages}{1182} (\bibinfo{year}{2022}), \eprint{https://www.science.org/doi/pdf/10.1126/science.abn7293}, \urlprefix\url{https://doi.org/10.1126/science.abn7293}.

\bibitem[{\citenamefont{Huang et~al.}(2020{\natexlab{a}})\citenamefont{Huang, Kueng, and Preskill}}]{huang2020predicting}
\bibinfo{author}{\bibfnamefont{H.-Y.} \bibnamefont{Huang}}, \bibinfo{author}{\bibfnamefont{R.}~\bibnamefont{Kueng}}, \bibnamefont{and} \bibinfo{author}{\bibfnamefont{J.}~\bibnamefont{Preskill}}, \bibinfo{journal}{Nature Physics} \textbf{\bibinfo{volume}{16}}, \bibinfo{pages}{1050} (\bibinfo{year}{2020}{\natexlab{a}}), ISSN \bibinfo{issn}{1745-2481}, \urlprefix\url{https://doi.org/10.1038/s41567-020-0932-7}.

\bibitem[{\citenamefont{Paris and Rehacek}(2004)}]{paris2004quantum}
\bibinfo{author}{\bibfnamefont{M.}~\bibnamefont{Paris}} \bibnamefont{and} \bibinfo{author}{\bibfnamefont{J.}~\bibnamefont{Rehacek}}, \emph{\bibinfo{title}{Quantum state estimation}}, vol. \bibinfo{volume}{649} (\bibinfo{publisher}{Springer Science \& Business Media}, \bibinfo{year}{2004}), \urlprefix\url{https://doi.org/10.1007/b98673}.

\bibitem[{\citenamefont{Adamson and Steinberg}(2010)}]{adamson2010improving}
\bibinfo{author}{\bibfnamefont{R.~B.~A.} \bibnamefont{Adamson}} \bibnamefont{and} \bibinfo{author}{\bibfnamefont{A.~M.} \bibnamefont{Steinberg}}, \bibinfo{journal}{Phys. Rev. Lett.} \textbf{\bibinfo{volume}{105}}, \bibinfo{pages}{030406} (\bibinfo{year}{2010}), \urlprefix\url{https://doi.org/10.1103/PhysRevLett.105.030406}.

\bibitem[{\citenamefont{Aaronson}(2018)}]{Aaronson2018shadow}
\bibinfo{author}{\bibfnamefont{S.}~\bibnamefont{Aaronson}}, in \emph{\bibinfo{booktitle}{Proceedings of the 50th Annual ACM SIGACT Symposium on Theory of Computing}} (\bibinfo{publisher}{Association for Computing Machinery}, \bibinfo{address}{New York, NY, USA}, \bibinfo{year}{2018}), STOC 2018, p. \bibinfo{pages}{325–338}, ISBN \bibinfo{isbn}{9781450355599}, \urlprefix\url{https://doi.org/10.1145/3188745.3188802}.

\bibitem[{\citenamefont{Wiebe et~al.}(2014)\citenamefont{Wiebe, Granade, Ferrie, and Cory}}]{wiebe2014hamiltonian}
\bibinfo{author}{\bibfnamefont{N.}~\bibnamefont{Wiebe}}, \bibinfo{author}{\bibfnamefont{C.}~\bibnamefont{Granade}}, \bibinfo{author}{\bibfnamefont{C.}~\bibnamefont{Ferrie}}, \bibnamefont{and} \bibinfo{author}{\bibfnamefont{D.~G.} \bibnamefont{Cory}}, \bibinfo{journal}{Phys. Rev. Lett.} \textbf{\bibinfo{volume}{112}}, \bibinfo{pages}{190501} (\bibinfo{year}{2014}), \urlprefix\url{https://doi.org/10.1103/PhysRevLett.112.190501}.

\bibitem[{\citenamefont{Flammia and Wallman}(2020)}]{flammia2020efficient}
\bibinfo{author}{\bibfnamefont{S.~T.} \bibnamefont{Flammia}} \bibnamefont{and} \bibinfo{author}{\bibfnamefont{J.~J.} \bibnamefont{Wallman}}, \bibinfo{journal}{ACM Transactions on Quantum Computing} \textbf{\bibinfo{volume}{1}} (\bibinfo{year}{2020}), \urlprefix\url{https://doi.org/10.1145/3408039}.

\bibitem[{\citenamefont{Flammia and O'Donnell}(2021)}]{flammia2021pauli}
\bibinfo{author}{\bibfnamefont{S.~T.} \bibnamefont{Flammia}} \bibnamefont{and} \bibinfo{author}{\bibfnamefont{R.}~\bibnamefont{O'Donnell}}, \bibinfo{journal}{{Quantum}} \textbf{\bibinfo{volume}{5}}, \bibinfo{pages}{549} (\bibinfo{year}{2021}), ISSN \bibinfo{issn}{2521-327X}, \urlprefix\url{https://doi.org/10.22331/q-2021-09-23-549}.

\bibitem[{\citenamefont{Urbina and Richter}(2013)}]{urbina2013random_state}
\bibinfo{author}{\bibfnamefont{J.-D.} \bibnamefont{Urbina}} \bibnamefont{and} \bibinfo{author}{\bibfnamefont{K.}~\bibnamefont{Richter}}, \bibinfo{journal}{Advances in Physics} \textbf{\bibinfo{volume}{62}}, \bibinfo{pages}{363} (\bibinfo{year}{2013}), \urlprefix\url{https://doi.org/10.1080/00018732.2013.860277}.

\bibitem[{\citenamefont{Emerson et~al.}(2005)\citenamefont{Emerson, Alicki, and Życzkowski}}]{RB2005}
\bibinfo{author}{\bibfnamefont{J.}~\bibnamefont{Emerson}}, \bibinfo{author}{\bibfnamefont{R.}~\bibnamefont{Alicki}}, \bibnamefont{and} \bibinfo{author}{\bibfnamefont{K.}~\bibnamefont{Życzkowski}}, \bibinfo{journal}{Journal of Optics B: Quantum and Semiclassical Optics} \textbf{\bibinfo{volume}{7}}, \bibinfo{pages}{S347} (\bibinfo{year}{2005}), \urlprefix\url{https://doi.org/10.1088/1464-4266/7/10/021}.

\bibitem[{\citenamefont{Knill et~al.}(2008)\citenamefont{Knill, Leibfried, Reichle, Britton, Blakestad, Jost, Langer, Ozeri, Seidelin, and Wineland}}]{RB2008}
\bibinfo{author}{\bibfnamefont{E.}~\bibnamefont{Knill}}, \bibinfo{author}{\bibfnamefont{D.}~\bibnamefont{Leibfried}}, \bibinfo{author}{\bibfnamefont{R.}~\bibnamefont{Reichle}}, \bibinfo{author}{\bibfnamefont{J.}~\bibnamefont{Britton}}, \bibinfo{author}{\bibfnamefont{R.~B.} \bibnamefont{Blakestad}}, \bibinfo{author}{\bibfnamefont{J.~D.} \bibnamefont{Jost}}, \bibinfo{author}{\bibfnamefont{C.}~\bibnamefont{Langer}}, \bibinfo{author}{\bibfnamefont{R.}~\bibnamefont{Ozeri}}, \bibinfo{author}{\bibfnamefont{S.}~\bibnamefont{Seidelin}}, \bibnamefont{and} \bibinfo{author}{\bibfnamefont{D.~J.} \bibnamefont{Wineland}}, \bibinfo{journal}{Phys. Rev. A} \textbf{\bibinfo{volume}{77}}, \bibinfo{pages}{012307} (\bibinfo{year}{2008}), \urlprefix\url{https://doi.org/10.1103/PhysRevA.77.012307}.

\bibitem[{\citenamefont{Huang et~al.}(2020{\natexlab{b}})\citenamefont{Huang, Kueng, and Preskill}}]{Classical_shadow}
\bibinfo{author}{\bibfnamefont{H.-Y.} \bibnamefont{Huang}}, \bibinfo{author}{\bibfnamefont{R.}~\bibnamefont{Kueng}}, \bibnamefont{and} \bibinfo{author}{\bibfnamefont{J.}~\bibnamefont{Preskill}}, \bibinfo{journal}{Nature Physics} \textbf{\bibinfo{volume}{16}}, \bibinfo{pages}{1050} (\bibinfo{year}{2020}{\natexlab{b}}), ISSN \bibinfo{issn}{1745-2481}, \urlprefix\url{https://doi.org/10.1038/s41567-020-0932-7}.

\bibitem[{\citenamefont{Elben et~al.}(2018)\citenamefont{Elben, Vermersch, Dalmonte, Cirac, and Zoller}}]{Elben2018Random}
\bibinfo{author}{\bibfnamefont{A.}~\bibnamefont{Elben}}, \bibinfo{author}{\bibfnamefont{B.}~\bibnamefont{Vermersch}}, \bibinfo{author}{\bibfnamefont{M.}~\bibnamefont{Dalmonte}}, \bibinfo{author}{\bibfnamefont{J.~I.} \bibnamefont{Cirac}}, \bibnamefont{and} \bibinfo{author}{\bibfnamefont{P.}~\bibnamefont{Zoller}}, \bibinfo{journal}{Phys. Rev. Lett.} \textbf{\bibinfo{volume}{120}}, \bibinfo{pages}{050406} (\bibinfo{year}{2018}), \urlprefix\url{https://doi.org/10.1103/PhysRevLett.120.050406}.

\bibitem[{\citenamefont{Brydges et~al.}(2019)\citenamefont{Brydges, Elben, Jurcevic, Vermersch, Maier, Lanyon, Zoller, Blatt, and Roos}}]{Tiff2019randomized}
\bibinfo{author}{\bibfnamefont{T.}~\bibnamefont{Brydges}}, \bibinfo{author}{\bibfnamefont{A.}~\bibnamefont{Elben}}, \bibinfo{author}{\bibfnamefont{P.}~\bibnamefont{Jurcevic}}, \bibinfo{author}{\bibfnamefont{B.}~\bibnamefont{Vermersch}}, \bibinfo{author}{\bibfnamefont{C.}~\bibnamefont{Maier}}, \bibinfo{author}{\bibfnamefont{B.~P.} \bibnamefont{Lanyon}}, \bibinfo{author}{\bibfnamefont{P.}~\bibnamefont{Zoller}}, \bibinfo{author}{\bibfnamefont{R.}~\bibnamefont{Blatt}}, \bibnamefont{and} \bibinfo{author}{\bibfnamefont{C.~F.} \bibnamefont{Roos}}, \bibinfo{journal}{Science} \textbf{\bibinfo{volume}{364}}, \bibinfo{pages}{260} (\bibinfo{year}{2019}), \eprint{https://www.science.org/doi/pdf/10.1126/science.aau4963}, \urlprefix\url{https://doi.org/10.1126/science.aau4963}.

\bibitem[{\citenamefont{Elben et~al.}(2023)\citenamefont{Elben, Flammia, Huang, Kueng, Preskill, Vermersch, and Zoller}}]{elben2023randomized_meas}
\bibinfo{author}{\bibfnamefont{A.}~\bibnamefont{Elben}}, \bibinfo{author}{\bibfnamefont{S.~T.} \bibnamefont{Flammia}}, \bibinfo{author}{\bibfnamefont{H.-Y.} \bibnamefont{Huang}}, \bibinfo{author}{\bibfnamefont{R.}~\bibnamefont{Kueng}}, \bibinfo{author}{\bibfnamefont{J.}~\bibnamefont{Preskill}}, \bibinfo{author}{\bibfnamefont{B.}~\bibnamefont{Vermersch}}, \bibnamefont{and} \bibinfo{author}{\bibfnamefont{P.}~\bibnamefont{Zoller}}, \bibinfo{journal}{Nature Reviews Physics} \textbf{\bibinfo{volume}{5}}, \bibinfo{pages}{9} (\bibinfo{year}{2023}), ISSN \bibinfo{issn}{2522-5820}, \urlprefix\url{https://doi.org/10.1038/s42254-022-00535-2}.

\bibitem[{\citenamefont{Huggins et~al.}(2022)\citenamefont{Huggins, O'Gorman, Rubin, Reichman, Babbush, and Lee}}]{Huggins2022}
\bibinfo{author}{\bibfnamefont{W.~J.} \bibnamefont{Huggins}}, \bibinfo{author}{\bibfnamefont{B.~A.} \bibnamefont{O'Gorman}}, \bibinfo{author}{\bibfnamefont{N.~C.} \bibnamefont{Rubin}}, \bibinfo{author}{\bibfnamefont{D.~R.} \bibnamefont{Reichman}}, \bibinfo{author}{\bibfnamefont{R.}~\bibnamefont{Babbush}}, \bibnamefont{and} \bibinfo{author}{\bibfnamefont{J.}~\bibnamefont{Lee}}, \bibinfo{journal}{Nature} \textbf{\bibinfo{volume}{603}}, \bibinfo{pages}{416} (\bibinfo{year}{2022}), ISSN \bibinfo{issn}{1476-4687}, \urlprefix\url{https://doi.org/10.1038/s41586-021-04351-z}.

\bibitem[{\citenamefont{Wallman and Flammia}(2014)}]{RB_with_confidence}
\bibinfo{author}{\bibfnamefont{J.~J.} \bibnamefont{Wallman}} \bibnamefont{and} \bibinfo{author}{\bibfnamefont{S.~T.} \bibnamefont{Flammia}}, \bibinfo{journal}{New Journal of Physics} \textbf{\bibinfo{volume}{16}}, \bibinfo{pages}{103032} (\bibinfo{year}{2014}), \urlprefix\url{https://doi.org/10.1088/1367-2630/16/10/103032}.

\bibitem[{\citenamefont{Helsen and Walter}(2023)}]{Thrifty_shadow}
\bibinfo{author}{\bibfnamefont{J.}~\bibnamefont{Helsen}} \bibnamefont{and} \bibinfo{author}{\bibfnamefont{M.}~\bibnamefont{Walter}}, \bibinfo{journal}{Phys. Rev. Lett.} \textbf{\bibinfo{volume}{131}}, \bibinfo{pages}{240602} (\bibinfo{year}{2023}), \urlprefix\url{https://doi.org/10.1103/PhysRevLett.131.240602}.

\bibitem[{\citenamefont{Zhou and Liu}(2023)}]{zhou2023performance}
\bibinfo{author}{\bibfnamefont{Y.}~\bibnamefont{Zhou}} \bibnamefont{and} \bibinfo{author}{\bibfnamefont{Q.}~\bibnamefont{Liu}}, \bibinfo{journal}{{Quantum}} \textbf{\bibinfo{volume}{7}}, \bibinfo{pages}{1044} (\bibinfo{year}{2023}), ISSN \bibinfo{issn}{2521-327X}, \urlprefix\url{https://doi.org/10.22331/q-2023-06-29-1044}.

\bibitem[{\citenamefont{Magesan et~al.}(2012)\citenamefont{Magesan, Gambetta, and Emerson}}]{RB2012}
\bibinfo{author}{\bibfnamefont{E.}~\bibnamefont{Magesan}}, \bibinfo{author}{\bibfnamefont{J.~M.} \bibnamefont{Gambetta}}, \bibnamefont{and} \bibinfo{author}{\bibfnamefont{J.}~\bibnamefont{Emerson}}, \bibinfo{journal}{Phys. Rev. A} \textbf{\bibinfo{volume}{85}}, \bibinfo{pages}{042311} (\bibinfo{year}{2012}), \urlprefix\url{https://doi.org/10.1103/PhysRevA.85.042311}.

\bibitem[{\citenamefont{da~Silva et~al.}(2011)\citenamefont{da~Silva, Landon-Cardinal, and Poulin}}]{da2011practical}
\bibinfo{author}{\bibfnamefont{M.~P.} \bibnamefont{da~Silva}}, \bibinfo{author}{\bibfnamefont{O.}~\bibnamefont{Landon-Cardinal}}, \bibnamefont{and} \bibinfo{author}{\bibfnamefont{D.}~\bibnamefont{Poulin}}, \bibinfo{journal}{Phys. Rev. Lett.} \textbf{\bibinfo{volume}{107}}, \bibinfo{pages}{210404} (\bibinfo{year}{2011}), \urlprefix\url{https://doi.org/10.1103/PhysRevLett.107.210404}.

\bibitem[{\citenamefont{Chuang and Nielsen}(1997)}]{chuang1997prescription}
\bibinfo{author}{\bibfnamefont{I.~L.} \bibnamefont{Chuang}} \bibnamefont{and} \bibinfo{author}{\bibfnamefont{M.~A.} \bibnamefont{Nielsen}}, \bibinfo{journal}{Journal of Modern Optics} \textbf{\bibinfo{volume}{44}}, \bibinfo{pages}{2455} (\bibinfo{year}{1997}), \eprint{https://www.tandfonline.com/doi/pdf/10.1080/09500349708231894}, \urlprefix\url{https://doi.org/10.1080/09500349708231894}.

\bibitem[{\citenamefont{Poyatos et~al.}(1997)\citenamefont{Poyatos, Cirac, and Zoller}}]{poyatos1997complete}
\bibinfo{author}{\bibfnamefont{J.~F.} \bibnamefont{Poyatos}}, \bibinfo{author}{\bibfnamefont{J.~I.} \bibnamefont{Cirac}}, \bibnamefont{and} \bibinfo{author}{\bibfnamefont{P.}~\bibnamefont{Zoller}}, \bibinfo{journal}{Phys. Rev. Lett.} \textbf{\bibinfo{volume}{78}}, \bibinfo{pages}{390} (\bibinfo{year}{1997}), \urlprefix\url{https://doi.org/10.1103/PhysRevLett.78.390}.

\bibitem[{\citenamefont{Cai et~al.}(2023)\citenamefont{Cai, Babbush, Benjamin, Endo, Huggins, Li, McClean, and O'Brien}}]{cai2023quantum_error_mitigation}
\bibinfo{author}{\bibfnamefont{Z.}~\bibnamefont{Cai}}, \bibinfo{author}{\bibfnamefont{R.}~\bibnamefont{Babbush}}, \bibinfo{author}{\bibfnamefont{S.~C.} \bibnamefont{Benjamin}}, \bibinfo{author}{\bibfnamefont{S.}~\bibnamefont{Endo}}, \bibinfo{author}{\bibfnamefont{W.~J.} \bibnamefont{Huggins}}, \bibinfo{author}{\bibfnamefont{Y.}~\bibnamefont{Li}}, \bibinfo{author}{\bibfnamefont{J.~R.} \bibnamefont{McClean}}, \bibnamefont{and} \bibinfo{author}{\bibfnamefont{T.~E.} \bibnamefont{O'Brien}}, \bibinfo{journal}{Rev. Mod. Phys.} \textbf{\bibinfo{volume}{95}}, \bibinfo{pages}{045005} (\bibinfo{year}{2023}), \urlprefix\url{https://doi.org/10.1103/RevModPhys.95.045005}.

\bibitem[{\citenamefont{Gühne and Tóth}(2009)}]{entanglement}
\bibinfo{author}{\bibfnamefont{O.}~\bibnamefont{Gühne}} \bibnamefont{and} \bibinfo{author}{\bibfnamefont{G.}~\bibnamefont{Tóth}}, \bibinfo{journal}{Physics Reports} \textbf{\bibinfo{volume}{474}}, \bibinfo{pages}{1} (\bibinfo{year}{2009}), ISSN \bibinfo{issn}{0370-1573}, \urlprefix\url{https://doi.org/10.1016/j.physrep.2009.02.004}.

\bibitem[{\citenamefont{Li et~al.}(2022)\citenamefont{Li, Fan, Gong, Ye, Chen, Wu, Guan, Deng, Rong, Huang et~al.}}]{Shaowei2022CZ}
\bibinfo{author}{\bibfnamefont{S.}~\bibnamefont{Li}}, \bibinfo{author}{\bibfnamefont{D.}~\bibnamefont{Fan}}, \bibinfo{author}{\bibfnamefont{M.}~\bibnamefont{Gong}}, \bibinfo{author}{\bibfnamefont{Y.}~\bibnamefont{Ye}}, \bibinfo{author}{\bibfnamefont{X.}~\bibnamefont{Chen}}, \bibinfo{author}{\bibfnamefont{Y.}~\bibnamefont{Wu}}, \bibinfo{author}{\bibfnamefont{H.}~\bibnamefont{Guan}}, \bibinfo{author}{\bibfnamefont{H.}~\bibnamefont{Deng}}, \bibinfo{author}{\bibfnamefont{H.}~\bibnamefont{Rong}}, \bibinfo{author}{\bibfnamefont{H.-L.} \bibnamefont{Huang}}, \bibnamefont{et~al.}, \bibinfo{journal}{Chin. Phys. Lett.} \textbf{\bibinfo{volume}{39}}, \bibinfo{pages}{030302} (\bibinfo{year}{2022}), \urlprefix\url{http://doi.org/10.1088/0256-307X/39/3/030302}.

\bibitem[{\citenamefont{Dankert et~al.}(2009)\citenamefont{Dankert, Cleve, Emerson, and Livine}}]{dankert2009exact}
\bibinfo{author}{\bibfnamefont{C.}~\bibnamefont{Dankert}}, \bibinfo{author}{\bibfnamefont{R.}~\bibnamefont{Cleve}}, \bibinfo{author}{\bibfnamefont{J.}~\bibnamefont{Emerson}}, \bibnamefont{and} \bibinfo{author}{\bibfnamefont{E.}~\bibnamefont{Livine}}, \bibinfo{journal}{Phys. Rev. A} \textbf{\bibinfo{volume}{80}}, \bibinfo{pages}{012304} (\bibinfo{year}{2009}), \urlprefix\url{https://doi.org/10.1103/PhysRevA.80.012304}.

\bibitem[{\citenamefont{Helsen et~al.}(2022)\citenamefont{Helsen, Roth, Onorati, Werner, and Eisert}}]{general}
\bibinfo{author}{\bibfnamefont{J.}~\bibnamefont{Helsen}}, \bibinfo{author}{\bibfnamefont{I.}~\bibnamefont{Roth}}, \bibinfo{author}{\bibfnamefont{E.}~\bibnamefont{Onorati}}, \bibinfo{author}{\bibfnamefont{A.}~\bibnamefont{Werner}}, \bibnamefont{and} \bibinfo{author}{\bibfnamefont{J.}~\bibnamefont{Eisert}}, \bibinfo{journal}{PRX Quantum} \textbf{\bibinfo{volume}{3}}, \bibinfo{pages}{020357} (\bibinfo{year}{2022}), \urlprefix\url{https://doi.org/10.1103/PRXQuantum.3.020357}.

\bibitem[{\citenamefont{Epstein et~al.}(2014)\citenamefont{Epstein, Cross, Magesan, and Gambetta}}]{epstein2014investigating}
\bibinfo{author}{\bibfnamefont{J.~M.} \bibnamefont{Epstein}}, \bibinfo{author}{\bibfnamefont{A.~W.} \bibnamefont{Cross}}, \bibinfo{author}{\bibfnamefont{E.}~\bibnamefont{Magesan}}, \bibnamefont{and} \bibinfo{author}{\bibfnamefont{J.~M.} \bibnamefont{Gambetta}}, \bibinfo{journal}{Phys. Rev. A} \textbf{\bibinfo{volume}{89}}, \bibinfo{pages}{062321} (\bibinfo{year}{2014}), \urlprefix\url{https://doi.org/10.1103/PhysRevA.89.062321}.

\bibitem[{\citenamefont{Helsen et~al.}(2018)\citenamefont{Helsen, Wallman, and Wehner}}]{multi-qubit_Clifford_representations}
\bibinfo{author}{\bibfnamefont{J.}~\bibnamefont{Helsen}}, \bibinfo{author}{\bibfnamefont{J.~J.} \bibnamefont{Wallman}}, \bibnamefont{and} \bibinfo{author}{\bibfnamefont{S.}~\bibnamefont{Wehner}}, \bibinfo{journal}{Journal of Mathematical Physics} \textbf{\bibinfo{volume}{59}}, \bibinfo{pages}{072201} (\bibinfo{year}{2018}), ISSN \bibinfo{issn}{0022-2488}, \eprint{https://pubs.aip.org/aip/jmp/article-pdf/doi/10.1063/1.4997688/15970851/072201\_1\_online.pdf}, \urlprefix\url{https://doi.org/10.1063/1.4997688}.

\bibitem[{\citenamefont{Helsen et~al.}(2019)\citenamefont{Helsen, Wallman, Flammia, and Wehner}}]{helsen2019UsingFewSamples}
\bibinfo{author}{\bibfnamefont{J.}~\bibnamefont{Helsen}}, \bibinfo{author}{\bibfnamefont{J.~J.} \bibnamefont{Wallman}}, \bibinfo{author}{\bibfnamefont{S.~T.} \bibnamefont{Flammia}}, \bibnamefont{and} \bibinfo{author}{\bibfnamefont{S.}~\bibnamefont{Wehner}}, \bibinfo{journal}{Phys. Rev. A} \textbf{\bibinfo{volume}{100}}, \bibinfo{pages}{032304} (\bibinfo{year}{2019}), \urlprefix\url{https://doi.org/10.1103/PhysRevA.100.032304}.

\bibitem[{\citenamefont{Wallman et~al.}(2015)\citenamefont{Wallman, Granade, Harper, and Flammia}}]{unitarity}
\bibinfo{author}{\bibfnamefont{J.}~\bibnamefont{Wallman}}, \bibinfo{author}{\bibfnamefont{C.}~\bibnamefont{Granade}}, \bibinfo{author}{\bibfnamefont{R.}~\bibnamefont{Harper}}, \bibnamefont{and} \bibinfo{author}{\bibfnamefont{S.~T.} \bibnamefont{Flammia}}, \bibinfo{journal}{New Journal of Physics} \textbf{\bibinfo{volume}{17}}, \bibinfo{pages}{113020} (\bibinfo{year}{2015}), \urlprefix\url{https://doi.org/10.1088/1367-2630/17/11/113020}.

\bibitem[{\citenamefont{Nielsen and Chuang}(2010)}]{nielsen2010quantum}
\bibinfo{author}{\bibfnamefont{M.~A.} \bibnamefont{Nielsen}} \bibnamefont{and} \bibinfo{author}{\bibfnamefont{I.~L.} \bibnamefont{Chuang}}, \emph{\bibinfo{title}{Quantum computation and quantum information}} (\bibinfo{publisher}{Cambridge university press}, \bibinfo{year}{2010}), \urlprefix\url{https://doi.org/10.1017/CBO9780511976667}.

\bibitem[{\citenamefont{Streltsov et~al.}(2015)\citenamefont{Streltsov, Singh, Dhar, Bera, and Adesso}}]{streltsov2015measuring}
\bibinfo{author}{\bibfnamefont{A.}~\bibnamefont{Streltsov}}, \bibinfo{author}{\bibfnamefont{U.}~\bibnamefont{Singh}}, \bibinfo{author}{\bibfnamefont{H.~S.} \bibnamefont{Dhar}}, \bibinfo{author}{\bibfnamefont{M.~N.} \bibnamefont{Bera}}, \bibnamefont{and} \bibinfo{author}{\bibfnamefont{G.}~\bibnamefont{Adesso}}, \bibinfo{journal}{Phys. Rev. Lett.} \textbf{\bibinfo{volume}{115}}, \bibinfo{pages}{020403} (\bibinfo{year}{2015}), \urlprefix\url{https://doi.org/10.1103/PhysRevLett.115.020403}.

\bibitem[{\citenamefont{Heinrich et~al.}(2023)\citenamefont{Heinrich, Kliesch, and Roth}}]{heinrich2022randomized}
\bibinfo{author}{\bibfnamefont{M.}~\bibnamefont{Heinrich}}, \bibinfo{author}{\bibfnamefont{M.}~\bibnamefont{Kliesch}}, \bibnamefont{and} \bibinfo{author}{\bibfnamefont{I.}~\bibnamefont{Roth}}, \emph{\bibinfo{title}{Randomized benchmarking with random quantum circuits}} (\bibinfo{year}{2023}), \eprint{2212.06181}, \urlprefix\url{https://doi.org/10.48550/arXiv.2212.06181}.

\end{thebibliography}

\end{document}